\documentclass[aps,prl,reprint,superscriptaddress]{revtex4-1}
\usepackage{amssymb,amsthm,amsmath,bbm,bbold,graphicx,epsfig,threeparttable,color}
\usepackage[colorlinks=true,linkcolor=blue]{hyperref}
\usepackage{ulem,enumerate}

\theoremstyle{definition}
\newtheorem{defi}{Definition}

\newtheorem{lem}[defi]{Lemma}

\begin{document}

\title{Quantifying the nonclassicality of pure dephasing}
\author{Hong-Bin Chen}
\email{hongbinchen@phys.ncku.edu.tw}
\affiliation{Department of Physics, National Cheng Kung University, Tainan 70101, Taiwan}
\affiliation{Center for Quantum Frontiers of Research \& Technology, NCKU, Tainan 70101, Taiwan}
\affiliation{Department of Engineering Science, National Cheng Kung University, Tainan 70101, Taiwan}
\author{Ping-Yuan Lo}
\affiliation{Department of Electrophysics, National Chiao Tung University, Hsinchu 30010, Taiwan}
\author{Clemens Gneiting}
\affiliation{Theoretical Quantum Physics Laboratory, RIKEN Cluster for Pioneering Research, Wako-shi, Saitama 351-0198, Japan}
\author{Joonwoo Bae}
\affiliation{School of Electrical Engineering, Korea Advanced Institute of Science and Technology (KAIST), 291 Daehak-ro, Yuseong-gu, Daejeon 34141, Republic of Korea}
\author{Yueh-Nan Chen}
\email{yuehnan@mail.ncku.edu.tw}
\affiliation{Department of Physics, National Cheng Kung University, Tainan 70101, Taiwan}
\affiliation{Center for Quantum Frontiers of Research \& Technology, NCKU, Tainan 70101, Taiwan}
\author{Franco Nori}
\affiliation{Theoretical Quantum Physics Laboratory, RIKEN Cluster for Pioneering Research, Wako-shi, Saitama 351-0198, Japan}
\affiliation{Physics Department, University of Michigan, Ann Arbor, Michigan 48109-1040, USA}
\date{\today}

\begin{abstract}
One of the central problems in quantum theory is to characterize, detect, and quantify quantumness in terms of classical strategies. Dephasing processes, caused by non-dissipative information exchange between quantum systems and environments, provides a natural platform for this purpose, as they control the quantum-to-classical transition. Recently, it has been shown that dephasing dynamics itself can exhibit (non)classical traits, depending on the nature of the system-environment correlations and the related (im)possibility to simulate these dynamics with Hamiltonian ensembles---the classical strategy. Here we establish the framework of detecting and quantifying the nonclassicality for pure dephasing dynamics. The uniqueness of the canonical representation of Hamiltonian ensembles is shown, and a constructive method to determine the latter is presented. We illustrate our method for qubit, qutrit, and qubit-pair pure dephasing and describe how to implement our approach with quantum process tomography experiments. Our work is readily applicable to present-day quantum experiments.
\newpage
\end{abstract}

\maketitle

The boundary between the quantum and the classical world has always been a fundamental issue in quantum mechanics
\cite{ballentine_q_statistics_rmp_1970,zurek_q-c_transition_rmp_2003,schlosshauer_q-c_transition_rmp_2005,modi_q-c_boundary_discord_rmp_2012}.
An operationally viable way to demonstrate the genuine quantum nature of an experiment relies on the impossibility to mimic certain statistical properties of interest by using a ``classical strategy''. According to this logic, the quantum nature of an experiment is only convincingly demonstrated if the experimental statistics cannot be mimicked by the classical strategy; thus excluding any loophole to explain the statistics with a classical model.

For example, under the assumptions of realism and locality, Bell \cite{bell_ineq_phys_1964} derived an inequality for correlations between the
statistics of measurements on a bipartite system. Whenever the inequality is violated, one cannot reproduce the correlations by using a local hidden variable model, the latter serving
as the classical strategy for mimicking the measurement statistics. Another important paradigm is the quantumness of a boson field, which is formulated in terms
of the Wigner function or the Glauber-Sudarshan $P$ representation \cite{wigner_func_pr_1932, glauber_pr_1963, sudarshan_prl_1963}. Whenever these functions
exhibit negative values, the classical explanation in terms of a probability distribution over phase space fails to represent the boson field.

Following this spirit, one may ask for a classical strategy to frame the ``quantumness'' of open system dynamics. This question has been addressed in different ways. In these
approaches, specific properties of system states, e.g., Wigner functions with negativities, violation of Leggett-Garg inequality, non-stochasticity of dynamical processes,
or detection of quantum coherence, are identified as indicators of nonclassicality and monitored during the temporal evolution \cite{rahimikeshari_process_n_cla_prl_2013,krishna_process_n_cla_pra_2016,neill_n_cla_trans_prl_2010,hengna_n_markovianity_scirep_2015,jenhsiang_process_n_cla_scirep_2017,
simon_kol_ext_theo_arxiv_2017,george_coh_wit_pra_2018,smirne_process_n_cla_qst_2019}.

Alternatively, we propose to take the presence or absence of quantum correlations between system and environment as a signature for the quantum nature of the open system dynamics. As was shown recently \cite{hongbin_process_n_cla_prl_2018}, such presence or absence of nonclassical system-environment correlations is intimately linked to the (im)possibility to simulate the open system dynamics with a Hamiltonian ensemble (HE), which may thus serve as the classical strategy to witness the nonclassicality of the open system dynamics. HEs, which are also used to describe disordered quantum systems, attribute to each member of a collection of (time-independent) Hamiltonians a probability of occurrence, giving rise to an effective average dynamics.

Finding a simulating HE certifies that the open dynamics is classical. The nonexistence of a simulating HE, on the other hand, can be proven by
the necessity to resort to a HE accompanied by negative quasi-distributions.
Although being conceptually clear, as was shown in Ref. \cite{hongbin_process_n_cla_prl_2018} for the example of an extended spin-boson model, this is technically
highly nontrival in general; especially for high dimensions.
For example, the closely related problem of random-unitary decomposition can in general merely be numerically implemented \cite{audenaert_randon_unitary_njp_2008}.
An efficient approach appears desirable.

On the other hand, analyzing dephasing is essential for the improvements of quantum information science and quantum technologies.
Besides its fundamental relevance for the quantum-to-classical transition \cite{zurek_q-c_transition_phys_tod_1991,joos_textbook_2003,schlosshauer_textbook_2007},
classicality of the dynamics, reflected by the existence of a simulating HE, can then be related to the in-principle possibility to correct errors caused by the HE \cite{buscemi_err_correction_prl_2005}. Furthermore, it also constitutes one of the main obstacles in the fabrication and manipulation of quantum information devices
\cite{vandersypen_nmr_nature_2001,petta_st0_qubit_science_2005,foletti_st0_qubit_nat_phys_2009,ladd_quant_comp_nature_2010,buluta_quant_comp_rpp_2011,georgescu_quant_simu_rmp_2014}. Different implementations for the simulation of controlled pure dephasing
\cite{myatt_simul_deph_nature_2000, liu_simul_deph_np_2011, liu_simul_deph_nc_2018} and its mitigation
\cite{veldhorst_qubit_control_nn_2014, shulman_supp_qubit_deph_nc_2014, delbecq_qubit_deph_prl_2016, balasubramanian_nv_centre_nm_2009, bar-gill_nv_centre_nc_2013} exist.
Other experiments highlight the potential of decoherence or pure dephasing to contribute positively to certain quantum information tasks, such as entanglement stabilization \cite{shankar_ent_sup_cond_nature_2013} or entanglement swap \cite{nakajima_elec_tran_deph_nc_2018}.

Here we introduce a measure of nonclassicality for pure dephasing dynamics, i.e., we focus on situations where dephasing constitutes the sole dynamical agent.
We begin with recasting any HE into a canonical form; within this framework, each HE is composed of the same canonical set of Hamiltonians, such that the accompanying (quasi-)distribution fully characterizes the HE. Let us remark that one can interpret the resulting representation as a random rotation model, since it is a (quasi-)distribution of rotations induced by the Hamiltonians. We also prove its existence and uniqueness. This promotes it to a faithful representation of the pure dephasing dynamics and allows us to unambiguously quantify the nonclassicality. Additonally, we outline a systematic procedure to retrieve (quasi-)distributions for pure dephasing and elaborate our ideas for qubit, qutrit, and qubit-pair examples.
Finally, we also discuss the implementation of our approach with quantum process tomography experiments to show the ready applicability
to present-day quantum experiments.
\vspace{12pt}

\noindent\textbf{\textsf{Results}}

\noindent\textbf{Averaged dynamics of Hamiltonian ensembles.} A HE $\{(p_\lambda,\widehat{H}_\lambda)\}_\lambda$ is a
collection of Hermitian operators $\widehat{H}_\lambda$ acting on the same system \cite{hongbin_process_n_cla_prl_2018,Kropf2016effective}, where each member Hamiltonian is drawn
according to the probability distribution $p_\lambda\geq0$. A system $\rho_0$, isolated from any environment, is sent into a unitarily-evolving channel
$\rho_\lambda(t)=\widehat{U}_\lambda\rho_0\widehat{U}_\lambda^\dagger$, with $\widehat{U}_\lambda=\exp[-i\widehat{H}_\lambda t/\hbar]$ for a chosen
$\widehat{H}_\lambda$ according to $p_\lambda$. Then, the dynamics of the averaged state $\overline{\rho}(t)$ is given by the unital map
 \begin{equation}
\overline{\rho}(t)=\mathcal{E}_t\{\rho_0\}=\int p_\lambda\widehat{U}_\lambda\rho_0\widehat{U}_\lambda^\dagger d\lambda.
\label{eq_ensembe_averaged_dynamics}
\end{equation}

Even though each single realization $\rho_\lambda(t)$ evolves unitarily, the averaged state $\overline{\rho}(t)$ exhibits incoherent behavior
\cite{Gneiting2016incoherent,Kropf2016effective,clemens_disordered_pra_2017,clemens_disordered_prl_2017}. A seminal and intriguing example is a single qubit subject to spectral disorder with HE given by
$\{(p(\omega),\hbar\omega\hat{\sigma}_{z}/2)\}_\omega$, then the averaged dynamics describes pure dephasing:
\begin{equation}
\bar{\rho}(t)=\left[
\begin{array}{cc}
\rho_{\uparrow\uparrow} & \rho_{\uparrow\downarrow}\phi(t) \\
\rho_{\downarrow\uparrow}\phi^\ast(t) & \rho_{\downarrow\downarrow}
\end{array}
\right], \label{eq_he_pure_dephasing}
\end{equation}
with the dephasing factor $\phi(t)=\int p(\omega)e^{-i\omega t}d\omega$ being the Fourier transform of the probability distribution $p(\omega)$.

The pure dephasing in Eq.~(\ref{eq_he_pure_dephasing}) is a consequence of the commuting member Hamiltonian $\hbar\omega\hat{\sigma}_{z}/2$ in the ensemble. Each
Hamiltonian induces a unitary rotation about the $z$-axis of the Bloch sphere at angular velocity $\omega$. This gives rise to an intuitive interpretation of pure
dephasing in terms of random phases: each component rotates at different angular velocity $\omega$ and hence possesses its own time-evolving phase. Consequently
the phase of the averaged system gradually blurs out.

Note that $p(\omega)$ is the probability distribution of the angular velocity and qualitatively characterizes the ``randomness'' of the random rotation. Whenever
$p(\omega)$ is specified, the dynamics is uniquely determined via the Fourier transform in Eq.~(\ref{eq_he_pure_dephasing}). This is also in line with our
classification of such pure dephasing as classical \cite{hongbin_process_n_cla_prl_2018} since it is a statistical mixture of rotations at different angular velocities.
Meanwhile, the experimental simulation of pure dephasing is implemented in a similar spirit \cite{myatt_simul_deph_nature_2000,liu_simul_deph_np_2011,liu_simul_deph_nc_2018}.

\noindent\textbf{Canonical Hamiltonian-ensemble representation.} Although $p(\omega)$ is particularly representative for characterizing qubit pure dephasing, it is obvious
that, in general cases with non-commuting or higher dimensional member Hamiltonians $\widehat{H}_\lambda$, the Fourier transform in
Eq.~(\ref{eq_he_pure_dephasing}) is not applicable. We are therefore spurred to explore the canonical Hamiltonian-ensemble representation (CHER) as
a generalized representation of an averaged dynamics.

To fully understand the CHER, we first observe that, since both $\widehat{H}_\lambda$ and density matrices
$\rho$ are Hermitian, they are elements in the Lie algebra $\mathfrak{u}(n)=\mathfrak{u}(1)\oplus\mathfrak{su}(n)$, which are spanned by the identity $\{\widehat{I}\}$
and $\{\widehat{L}_m\}_m$ of $n^2-1$ traceless Hermitian generators, respectively. Then $\widehat{H}_\lambda\in\mathfrak{u}(n)$ is a linear combination
$\widehat{H}_\lambda=\lambda_0\widehat{I}+\sum_{m=1}^{n^2-1} \lambda_m\widehat{L}_m=\lambda_0\widehat{I}+\boldsymbol{\lambda}\cdot\widehat{\mathbf{L}}$,
where $\lambda_0\in\mathbb{R}$ and $\boldsymbol{\lambda}=\{\lambda_m\}_m\in\mathbb{R}^{n^2-1}$.

Since the dynamics is a linear map acting on $\rho$, invoking to the adjoint representation (see \textbf{\textsf{Methods}} and Supplementary Note 1), we can
assign each generator $\widehat{L}_m$ a linear map $\widehat{L}_m\mapsto\widetilde{L}_m\in\mathfrak{gl}(\mathfrak{u}(n))$, with its action
$\widetilde{L}_m(\bullet)=[\widehat{L}_m,\bullet]$ defined in terms of the commutator.

With the above mathematical setup, given a HE $\{(p_\lambda,\widehat{H}_\lambda)\}_\lambda$, one can consider the probability distribution $p_\lambda$ as a
CHER of an averaged dynamics $\mathcal{E}_t$, in the sense that Eq.~(\ref{eq_ensembe_averaged_dynamics}) can always be recast into
a Fourier transform from $p_\lambda$, on a locally compact group $\mathcal{G}$ characterized by
the parameter space $\lambda=\{\lambda_0,\boldsymbol{\lambda}\}$, to the dynamical linear map $\mathcal{E}^{(\widetilde{L})}_t$:
\begin{equation}
\mathcal{E}^{(\widetilde{L})}_t=\int_\mathcal{G} p_\lambda e^{-i\lambda\widetilde{L}t} d\lambda.
\label{eq_general_f-s_transform}
\end{equation}
Note that we have set $\hbar=1$ for symbolic abbreviation. Similarly, we can also express
$\rho=n^{-1}\widehat{I}+\boldsymbol{\rho}\cdot\widehat{\mathbf{L}}$ in terms of a column vector $\rho=\{n^{-1},\boldsymbol{\rho}\}$,
the action of $\mathcal{E}_t$ on $\rho$ is then the usual matrix multiplication
$\mathcal{E}_t\{\rho\}=\mathcal{E}^{(\widetilde{L})}_t\cdot\rho$ [see Supplementary Note 2 for the proof of Eq.~(\ref{eq_general_f-s_transform})].

We emphasize that, compared with Eq.~(\ref{eq_ensembe_averaged_dynamics}), the Fourier transform formalism~(\ref{eq_general_f-s_transform})
is a powerful tool in the following proof of uniqueness and establishment of our procedure. It also highlights our exclusive focus on the dynamics alone, regardless of
the system state. Additionally, it provides further insights into the nature of CHER and the connection to the process nonclassicality, in terms of a random rotation model.
In such interpretation, different components rotate about different axes, defined by the generators $\{\widehat{L}_m\}_m$.
Moreover, $p_\lambda$ is the distribution function of the random rotations over the $n^2$-dimensional Euclidean space.
This interpretation is consistent with the random phase model in the case of qubit pure dephasing~(\ref{eq_he_pure_dephasing}).

\noindent\textbf{HE simulation and process nonclassicality.} So far we have discussed the averaged dynamics of an isolated
system, in the absence of any environment, governed by a HE. Conversely, to discuss the nonclassicality of an open system dynamics reduced from a
system-environment arrangement, we should construct a simulating $\{(\wp_\lambda,\widehat{H}_\lambda)\}_\lambda$ for a given unital dynamics.

An autonomous system-environment arrangement is characterized by a time-independent total Hamiltonian $\widehat{H}_\mathrm{T}$ and evolves unitarily with
$\widehat{U}_\mathrm{T}=\exp[-i\widehat{H}_\mathrm{T}t]$. We have shown that \cite{hongbin_process_n_cla_prl_2018}, if the total system
$\rho_\mathrm{T}(t)=\widehat{U}_\mathrm{T}\rho_\mathrm{T}(0)\widehat{U}_\mathrm{T}^\dagger$ remains at all times classically correlated between the system and
its environment, displaying neither quantum discord \cite{quantum_discord_prl_2002,condition_nonzero_discord_prl_2010} nor entanglement, then the reduced system
dynamics $\rho_\mathrm{S}(t)=\mathcal{E}_t\{\rho_\mathrm{S}(0)\}=\mathrm{Tr}_\mathrm{E}[\rho_\mathrm{T}(t)]$ can be described by a time-independent HE equipped
with a legitimate (i.e., non-negative and normalized to unity) probability distribution. Moreover, such ensemble description under classical environments
in the absence of back-action has also been discussed in the literature \cite{franco_ensemble_pra_2012,franco_ensemble_nat_commun_2013,franco_ensemble_ann_phys_2014}.

However, given exclusively the knowledge on the reduced system dynamics $\mathcal{E}_t$, it is impossible to fully verify the correlations between the system and
its environment. Counter-intuitively, even if we have limited access to the system alone, the emergence of nonclassical correlations can be witnessed, whenever
one has no way to simulate the dynamics with any HEs equipped with a legitimate probability distribution. Such impossibility to simulate arises from the buildup
of nonclassical correlations. On the other hand, if such simulation is possible, one can explain $\mathcal{E}_t$ as a classical random rotation model. We therefore
define the negative values of the quasi-distribution $\wp_\lambda$ within the simulating HE as an indicator of process nonclassicality \cite{hongbin_process_n_cla_prl_2018}.

\noindent\textbf{Existence and uniqueness of the CHER for pure dephasing.} Here we promote the $\wp_\lambda$ within the simulating HE as a
CHER for a reduced system dynamics. In particular, by further investigating the underlying algebraic structures, we can show that such CHER for pure dephasing
is even faithful, provided diagonal member Hamiltonians. More precisely, for any pure dephasing dynamics, there always exists a unique simulating HE of
diagonal member Hamiltonians, equipped with either a legitimate or quasi-distribution.

The proof will become intelligible only after introducing our procedure to find the CHER below. We postpone it to Supplementary Note 8.

Since $\wp_\lambda$ is a distribution function over the parameter space of diagonal member Hamiltonians, along with the Fourier transform on the group $\mathcal{G}$ in
Eq.~(\ref{eq_general_f-s_transform}), this endows the CHER with a geometric interpretation of pure dephasing in terms of random rotation model.
Consequently, the CHER is particularly competent in characterization of the nonclassicality of pure dephasing.

\noindent\textbf{The nonclassicality measure for pure dephasing dynamics.} Having characterized the HE simulation of pure dephasing and its representation, we are now ready to propose the measure of nonclassicality of dynamics. The measure aims to provide an operational quantification on the nonclassicality of a pure dephasing dynamics. Due to the existence and uniqueness, every pure dephasing $\mathcal{E}_t$ can be assigned a unique (quasi-)distribution $\wp_\lambda$. We emphasize that it is the distribution $\wp_{\lambda}$ which gives the characterization of the nonclassicality: unless they correspond to legitimate probabilities, no HE exists for the exact simulation.

The nonclassicality measure for a dynamics $\mathcal{E}_t$ assigned with a unique (quasi-)distribution $\wp_\lambda$ is as follows,
\begin{equation}
\mathcal{N}\{\mathcal{E}_t\} = \inf_{p_\lambda} \mathrm{D}(\wp_\lambda, p_\lambda ),~\mathrm{with}~ \mathrm{D}(p_\lambda,p_\lambda^\prime)=\int_\mathcal{G}\frac{1}{2}\vert p_\lambda-p_\lambda^\prime\vert d\lambda, \label{eq_measure_nonclassicality}
\end{equation}
where the infimum runs over all classical probability distributions $p_\lambda$ over the parameter space $\mathcal{G}$ of the diagonal member Hamiltonians. The variational distance $\mathrm{D}(p_\lambda,p_\lambda^\prime)$ has an operational meaning as the single-shot distinguishability: it quantifies the highest success probability of distinguishing two probabilistic systems $p_\lambda$ and $p_\lambda^\prime$, such that $p_{\mathrm{success}} = [1 + \mathrm{D}(p_\lambda,p_\lambda^\prime)]/2$.

The measure proposed in Eq. (\ref{eq_measure_nonclassicality}) contains advantages and useful properties for the quantification. First, the measure has a clear operational meaning. It tells how well a dynamics $\mathcal{E}_t$ can be simulated by a HE. The possibility of making success or failure in the simulation with a HE can be found. Second, the measure is monotonic that the larger it is, the harder a classical simulation is. This follows from the fact that the classical dynamics of pure dephasing forms a convex set, i.e., their probabilistic mixture is also classical. The proof is presented in Supplementary Note 3. It is noteworthy that the convexity can be constructed by considering (quasi-)probabilities of dynamics, but not dynamics {\it per se}. Finally, we also note that the measure shares some similarities with the quantification of non-Markovianity \cite{BLP_measure_prl_2009}.

In what follows, we consider the nonclassicality of pure dephasing dynamics on a single qubit reduced from the extended spin-boson model \cite{hongbin_process_n_cla_prl_2018} with
a relative phase between the coupling constants, i.e., $g_{2,\mathbf{k}}=g_{1,\mathbf{k}}e^{i\varphi}$. The quasi-distribution $\wp_\mathrm{o1}^{(\mathrm{X})}(\omega)$
represents the single qubit pure dephasing and, consequently, its nonclassicality varies with $\varphi$. The results are shown in Fig.~1.

\begin{figure}[th]
\includegraphics[width=\columnwidth]{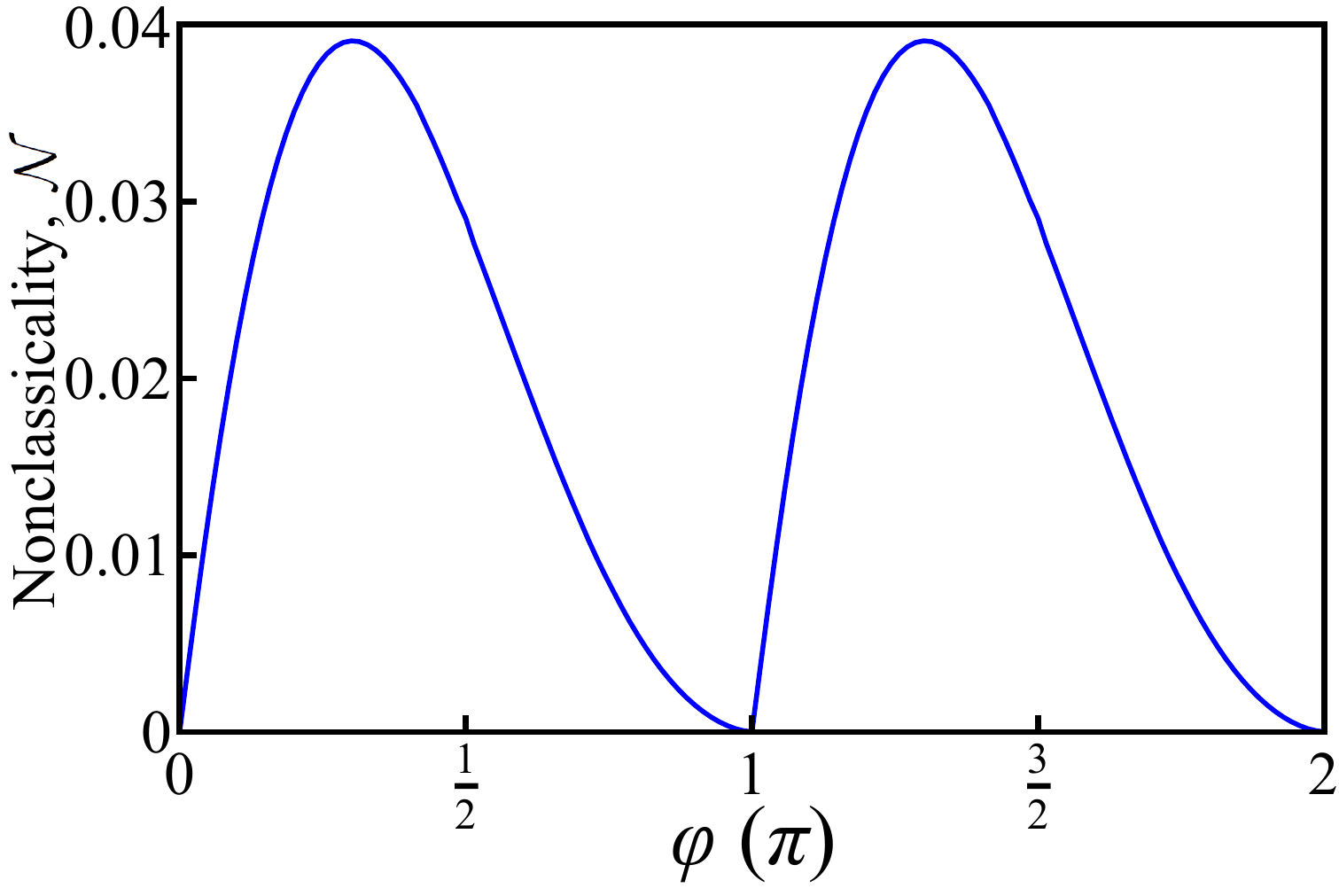}
\caption{The nonclassicality of the qubit pure dephasing. We consider the qubit pure dephasing reduced from the extended spin-boson model, wherein $\varphi$
(in unit of $\pi$) is the relative phase between the coupling constants of the qubit-pair to the common boson environment. The nonclassicality $\mathcal{N}$
is quantified according to Eq.~(\ref{eq_measure_nonclassicality}). In this example, the Ohmic spectral density
$\mathcal{J}(\omega)=\omega\exp(-\omega/\omega_\mathrm{c})$ with cut-off $\omega_\mathrm{c}=1$ and the zero-temperature limit are considered.}
\label{fig_pure dephasing nonclassicality}
\end{figure}

\noindent\textbf{Retrieval of the (quasi-)distribution.} Given a HE, it is, in principle, straightforward to calculate the averaged dynamics of an isolated system,
according to Eq.~(\ref{eq_ensembe_averaged_dynamics}) [or, equivalently, to Eq.~(\ref{eq_general_f-s_transform})]. Nevertheless, to find the solution to the
inverse transform of Eq.~(\ref{eq_general_f-s_transform}), i.e., retrieval of the (quasi-)distribution within the simulating HE for a given reduced dynamics,
is formidable in general, in contrast to the conventional inverse Fourier transform.
Consequently, to establish a systematic procedure to find the CHER of pure dephasing dynamics is very desirable.

In view of the qubit pure dephasing in Eq.~(\ref{eq_he_pure_dephasing}), to simulate any higher dimensional pure dephasing dynamics, we focus on the traceless
and diagonal member Hamiltonian such that $\widehat{H}_\lambda=\lambda\widehat{L}$ belongs to the Cartan subalgebra (CSA) $\mathfrak{H}$ of
$\mathfrak{su}(n)$ (see \textbf{\textsf{Methods}}). The tracelessness is due to the fact that the trace plays no role in describing the dynamics. Additionally,
since the adjoint representation preserves the structure of commutator, the adjoint representation of $\mathfrak{H}$ is also a CSA of
$\mathfrak{sl}(\mathfrak{u}(n))$. We therefore have the following commutativity $[\lambda\widehat{L},\lambda^\prime\widehat{L}]=0\Leftrightarrow
[\lambda\widetilde{L},\lambda^\prime\widetilde{L}]=0$.

It should be noted that, even if $\lambda\widehat{L}\in\mathfrak{H}$ can be chosen to be diagonal, $\lambda\widetilde{L}$ itself may not necessarily be
diagonal as well since the generators of $\mathfrak{u}(n)$ are not the suitable bases for diagonalizing it. As we will see below, the diagonalization of the adjoint
representation is a critical step to the retrieval of the (quasi-)distribution for pure dephasing.

Furthermore, the conventional inverse Fourier transform does not work because we are now dealing with linear maps in the $\mathfrak{sl}(\mathfrak{u}(n))$
space. To efficiently establish a set of equations governing the CHER of pure dephasing, we inevitably encounter increasingly many mathematical
terminologies, especially those specifying the intrinsic algebraic structures within the CHER. To make our procedure transparent, we instead
demonstrate several examples, each of which reveals the central concepts of our procedure, rather than elaborate the mathematical tutorial.
Our approach can be easily generalized to higher dimensional pure dephasing.

\noindent\textbf{Procedure towards the CHER of pure dephasing.} We begin with the case of qubit pure dephasing. Although this problem has been
discussed \cite{hongbin_process_n_cla_prl_2018}, it relies on the conventional Fourier transform and Bochner's theorem \cite{bochner_math_ann_1933} and cannot be
generalized to higher dimensional systems. Here we recast it into Eq.~(\ref{eq_general_f-s_transform}). This helps us to establish a systematic
procedure for higher dimensional problems.

Within a properly chosen basis, a qubit pure dephasing, reduced from a system-environment arrangement, can be expressed in the same form as
Eq.~(\ref{eq_he_pure_dephasing}). Unlike the one resulting from ensemble average, the dephasing factor $\phi(t)=\exp[-i\theta(t)-\Phi(t)]$ is determined by the
system-environment interaction, where $\theta(t)$ ($\Phi(t)$) is a real odd (even) function on time $t$, respectively, such that $\phi(0)=1$,
$\vert\phi(t)\vert\leq1$, and $\phi(-t)=\phi^\ast(t)$. The dynamical linear map $\mathcal{E}^{(\tilde{\sigma})}_t$ can be constructed by applying
$\mathcal{E}_t\{\hat{\sigma}_m\}=\sum_{l=0}^{3}\hat{\sigma}_l [\mathcal{E}^{(\tilde{\sigma})}_t]_{lm}$ on each generator, where
$\hat{\sigma}_0=\widehat{I}$ is the identity and $\hat{\sigma}_{1,2,3}$ denotes the three Pauli matrices.

To find the CHER, we mean to find a (quasi-)distribution $\wp(\omega)$ encapsulated within the simulating HE
$\{(\wp(\omega),\widehat{H}_\omega=\omega\hat{\sigma}_z/2)\}_\omega$ satisfying
\begin{equation}
\mathcal{E}^{(\widetilde{\sigma})}_t=\int_\mathbb{R}\wp(\omega)e^{-i(\omega\tilde{\sigma}_z/2)t}d\omega.
\label{eq_qubit_deph_fst}
\end{equation}
The same conclusion $\exp[-i\theta(t)-\Phi(t)]=\int_\mathbb{R}\wp(\omega)e^{-i\omega t}d\omega$ is easily seen after diagonalizing Eq.~(\ref{eq_qubit_deph_fst})
(see Supplementary Note 4 for more details). Finally, performing the conventional inverse Fourier transform leads to the desired result $\wp(\omega)$.

To understand the deeper insight behind the diagonalization, we observe that the diagonalization changes the basis from the three pauli matrices into raising
and lowering operators and leaves $\hat{\sigma}_z$ invariant; namely, $\{\hat{\sigma}_+,\hat{\sigma}_-,\hat{\sigma}_z\}$, which are the generators of
$\mathfrak{sl}(2)$. In other words, they are the common ``eigenvectors'' of $\widehat{H}_\omega$ with ``eigenvalues'' $\pm1$ in the sense of the adjoint
representation, $\widetilde{H}_\omega(\hat{\sigma}_\pm)=[\omega\hat{\sigma}_z/2,\hat{\sigma}_\pm]=\pm1\cdot\omega\hat{\sigma}_\pm$ (see Supplementary Note 5 for
more details). The eigenvalues $\pm1$ are referred to as the roots (denoted by $\alpha_{1,2}$) associated to the root spaces $\mathrm{span}\{\hat{\sigma}_\pm\}$,
spanned by the operators $\hat{\sigma}_\pm$, respectively. However, for higher dimensional systems, the roots are no longer real scalars but vectors in an Euclidean
space. This can be better seen as follow.

A qutrit pure dephasing can be written as
\begin{equation}
\rho(t)=\mathcal{E}_t\{\rho_0\}=\left[
\begin{array}{ccc}
\rho_{11} & \rho_{12}\phi_1(t) & \rho_{13}\phi_4(t) \\
\rho_{21}\phi_2(t) & \rho_{22} & \rho_{23}\phi_6(t) \\
\rho_{31}\phi_5(t) & \rho_{32}\phi_7(t) & \rho_{33}
\end{array}
\right],
\end{equation}
To guarantee the Hermicity of $\rho(t)$, the dephasing factors must further satisfy $\phi_1(t)=\phi_2^\ast(t)$, and so on.

To expand $\rho$ as a nine-dimensional column vector, it is natural to use the Gell-Mann matrices (denoted by $\hat{\sigma}_m$, $m=1,\ldots,8$)
as the generators of $\mathfrak{su}(3)$. However, after
the diagonalization, the basis is changed into that of $\mathfrak{gl}(3)$ (e.g., $\widehat{K}_0=\widehat{I}$,
$\widehat{K}_1=\widehat{K}_2^\dagger=(\hat{\sigma}_1+i\hat{\sigma}_2)/2$, and $\widehat{K}_3=\widehat{L}_3=\hat{\sigma}_3$). Within this basis, the dynamical linear map
$\mathcal{E}^{(\widetilde{L})}_t$ is diagonalized, i.e., $\mathcal{E}_t\{\widehat{K}_m\}=\widehat{K}_m\phi_m(t)$.

In this case, we consider the member Hamiltonian $\widehat{H}_{\boldsymbol{\lambda}}=(\lambda_3\widehat{L}_3+\lambda_8\widehat{L}_8)/2\in\mathfrak{H}$
and $\boldsymbol{\lambda}=(\lambda_3,\lambda_8)\in\mathbb{R}^2$. After estimating all the commutators
$[\widehat{H}_{\boldsymbol{\lambda}},\widehat{K}_m]=(\boldsymbol{\alpha}_m\cdot\boldsymbol{\lambda})\widehat{K}_m$, we obtain its adjoint representation
$\widetilde{H}_{\boldsymbol{\lambda}}=(\lambda_3\widetilde{L}_3+\lambda_8\widetilde{L}_8)/2$, which is diagonal in the $\mathfrak{gl}(3)$ basis.

Finally, according to Eq.~(\ref{eq_general_f-s_transform}) $\mathcal{E}^{(\widetilde{L})}_t=\int_{\mathbb{R}^2}\wp(\lambda_3,\lambda_8)e^{-i\widetilde{H}_{\vec{\lambda}}t} d\lambda_3d\lambda_8$, we conclude that the
(quasi-)distribution $\wp(\lambda_3,\lambda_8)$ is governed by the following simultaneous Fourier transforms:
\begin{eqnarray}
\phi_1(t)&=&\int_{\mathbb{R}^2}\wp(\lambda_3,\lambda_8)e^{-i(\boldsymbol{\alpha}_1\cdot\boldsymbol{\lambda})t}d\lambda_3d\lambda_8 \label{eq_simu_qutrit_pure_deph_1}, \\
\phi_4(t)&=&\int_{\mathbb{R}^2}\wp(\lambda_3,\lambda_8)e^{-i(\boldsymbol{\alpha}_4\cdot\boldsymbol{\lambda})t}d\lambda_3d\lambda_8 \label{eq_simu_qutrit_pure_deph_4}, \\
\phi_6(t)&=&\int_{\mathbb{R}^2}\wp(\lambda_3,\lambda_8)e^{-i(\boldsymbol{\alpha}_6\cdot\boldsymbol{\lambda})t}d\lambda_3d\lambda_8 \label{eq_simu_qutrit_pure_deph_6}.
\end{eqnarray}

We can collect the six non-zero root vectors $\boldsymbol{\alpha}_m$. They are two-dimensional vectors of equal length in the $\lambda_3$-$\lambda_8$ plane forming
the root system R of $\mathfrak{su}(3)$. We plot them in Fig.~2. Further details are given in Supplementary Note 6.

Similarly, for $n$-dimensional pure dephasing, each member Hamiltonian $\widehat{H}_{\boldsymbol{\lambda}}$, taken from the $\mathfrak{H}$ of $\mathfrak{su}(n)$, possesses $n-1$ free parameters
$\boldsymbol{\lambda}=\{\lambda_{k^2-1}\}_{k=2,3,\ldots,n}$; meanwhile, the (quasi-)distribution $\wp(\boldsymbol{\lambda})$ is defined on the $(n-1)$-dimensional Euclidean space. Moreover, the action
of $\widehat{H}_{\boldsymbol{\lambda}}$ on the $n^2-n$ root spaces $\mathrm{span}\{\widehat{K}_m\}$ is described by the root system $\mathrm{R}=\{\boldsymbol{\alpha}_m\}_m$, consisting of $n^2-n$ real
vectors of $(n-1)$-dimension. Further properties of R reduce the complexity of our procedure (see \textbf{\textsf{Methods}}).

Consequently, combining the techniques, i.e., the adjoint representation, the Fourier transform on groups, and the root space decomposition, we can concisely formulate our procedure to find the CHER
$\wp(\boldsymbol{\lambda})$ for the $n$-dimensional pure dephasing. We restrict ourselves to the diagonal member Hamiltonians (in $\mathfrak{H}$) and establish its root system R. The
(quasi-)distribution $\wp(\boldsymbol{\lambda})$ is characterized by the $(n^2-n)/2$ Fourier transforms with respect to positive roots and its corresponding dephasing factor $\phi_m(t)$ associated to
the root space $\mathrm{span}\{\widehat{K}_m\}$:
\begin{equation}
\phi_m(t)=\int_{\mathbb{R}^{n-1}}\wp(\boldsymbol{\lambda})e^{-i(\boldsymbol{\alpha}_m\cdot\boldsymbol{\lambda})t}d^{n-1}\boldsymbol{\lambda},\mathrm{for~positive~roots}~\boldsymbol{\alpha}_m.
\label{eq_pure_deph_f-s_transform}
\end{equation}
Furthermore, the simple roots define a new set of random variables $x_m=\boldsymbol{\alpha}_m\cdot\boldsymbol{\lambda}$, for simple roots $\boldsymbol{\alpha}_m$, and their corresponding equations
define the marginals of $\wp(\boldsymbol{\lambda})$ along $x_m$. The other equations describe the correlations among $x_m$.

\begin{figure}[th]
\includegraphics[width=\columnwidth]{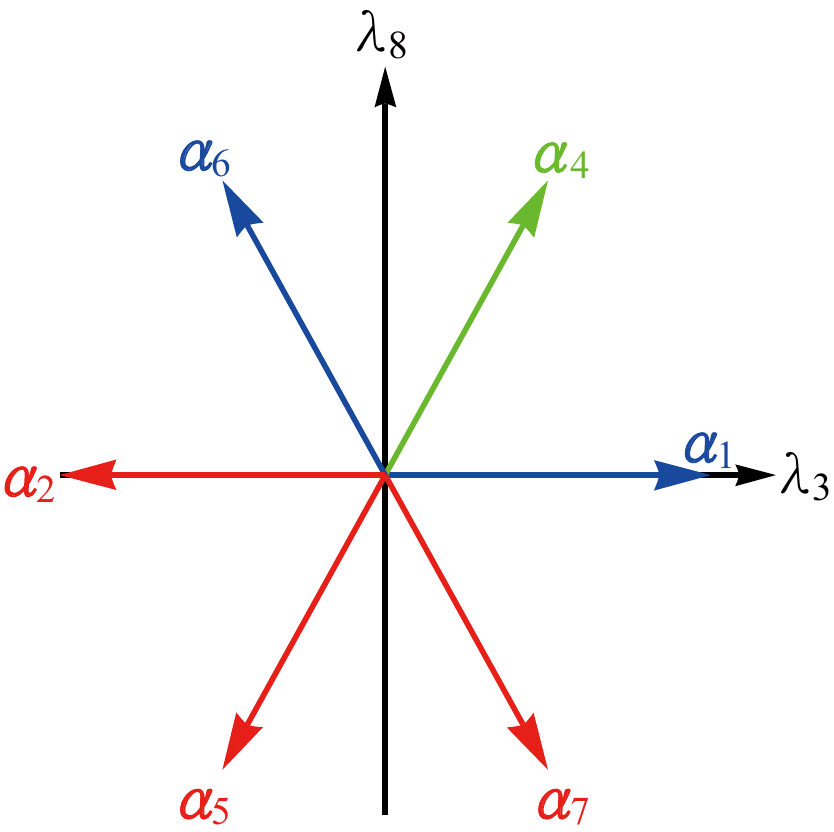}
\caption{The root system R of $\mathfrak{su}(3)$. It consists of six non-zero root vectors on the $\lambda_3$-$\lambda_8$ plane. Among them, $\boldsymbol{\alpha}_1$ (blue), $\boldsymbol{\alpha}_4$
(green), and $\boldsymbol{\alpha}_6$ (blue) are positive and the other three (red) are negative since roots are always come in pair with opposite directions.
Also, $\boldsymbol{\alpha}_1$ and $\boldsymbol{\alpha}_6$ are simple because $\boldsymbol{\alpha}_4=\boldsymbol{\alpha}_1+\boldsymbol{\alpha}_6$ is a combination of simple roots.}
\label{fig_roots-su(3)}
\end{figure}

\noindent\textbf{Example: qubit pair pure dephasing.} As an instructive paradigm demonstrating our procedure to find the CHER of pure dephasing, we consider
the extended spin-boson model consisting of a non-interacting qubit pair coupled to a common boson bath (Fig.~3\textbf{a}) with total Hamiltonian
$\widehat{H}_\mathrm{T}=\sum_{j=1,2}\omega_j\hat{\sigma}_{z,j}/2+\sum_{\mathbf{k}}\omega_{\mathbf{k}}\hat{b}_{\mathbf{k}}^\dagger\hat{b}_{\mathbf{k}}
+\sum_{j,\mathbf{k}}\hat{\sigma}_{z,j}\otimes(g_{j,\mathbf{k}}\hat{b}_{\mathbf{k}}^\dagger+g_{j,\mathbf{k}}^\ast\hat{b}_{\mathbf{k}})$. We now focus on the pure dephasing of
the qubit pair as a $4\times4$ system. The full dynamics has been given in Ref.~\cite{hongbin_process_n_cla_prl_2018}.

To simulate the qubit pair pure dephasing, the diagonal member Hamiltonian is taken from the $\mathfrak{H}$ of $\mathfrak{su}(4)$
$\widehat{H}_{\boldsymbol{\lambda}}=(\lambda_3\widehat{L}_3+\lambda_8\widehat{L}_8+\lambda_{15}\widehat{L}_{15})/2$ and $\wp(\boldsymbol{\lambda})$ is a (quasi-)distribution
over $\mathbb{R}^3$ space with $\lambda_3$, $\lambda_8$, and $\lambda_{15}$ being its axes. Note that the $\mathfrak{su}(4)$ has six
positive root vectors and three among them are simple, and all positive root vectors can be obtained by combining simple ones (Fig.~3\textbf{b}). We
perform the change of variables $x_m=\boldsymbol{\alpha}_m\cdot\boldsymbol{\lambda}$, $m=1,6,13$. Then, the (quasi-)distribution changes as
$\wp(\boldsymbol{\lambda})\mapsto\wp^\prime(x_1,x_6,x_{13})$. The three axes of $\wp^\prime(\mathbf{x})$ are defined by the three simple root vectors.

Additionally, since $\phi_6(t)=1$, by observing the special correspondences between root vectors and dephasing factors, we can assume that
\begin{equation}
\wp^\prime(\mathbf{x})=\wp_6(x_6)\wp_{1,13}(x_1,x_{13})
\end{equation}
is separable into two parties. The Fourier equation for $\phi_6(t)$ leads to the result that $\wp_6(x_6)=\delta(x_6)$ and those for $\phi_1(t)$ and
$\phi_{13}(t)$ specify the marginals of $\wp_{1,13}(x_1,x_{13})$ along the direction $\boldsymbol{\alpha}_1$ and $\boldsymbol{\alpha}_{13}$, respectively; meanwhile the one for
$\phi_9(t)$
\begin{equation}
\phi_9(t)=\int_{\mathbb{R}^2}\wp_{1,13}(x_1,x_{13})e^{-ix_1 t}e^{-ix_{13} t}dx_1dx_{13}
\label{eq_qubit_pair_depa_corr}
\end{equation}
describes the correlation between $x_1$ and $x_{13}$.

For the case of Ohmic
spectral density $\mathcal{J}(\omega)=\omega\exp(-\omega/\omega_\mathrm{c})$ in the zero-temperature limit, Eq.~(\ref{eq_qubit_pair_depa_corr}) can be recast into
a conventional two-dimensional Fourier transform by a simple ansatz. Then, $\wp_{1,13}(x_1,x_{13})$ can be easily obtained by conventional inverse transform and
the numerical result is shown in Fig.~3\textbf{c}. It exhibits manifest negative regions and illustrates the nonclassical nature of the qubit pair pure
dephasing. Detailed calculations are given in Supplementary Note 7.

Finally, having introducing our procedure to find the CHER, we combine it with the investigation on the intrinsic algebraic structure. Then
the uniqueness of the CHER for pure dephasing is intelligible and the detailed proof is given in Supplementary Note 8.

It is worthwhile to recall that similar models, in which several qubits were coupled identically to a common bath, had been considered
\cite{palma_dfs_prslsa_1996,duan_dfs_prl_1997,zanardi_dfs_prl_1997}, wherein the suppression of decoherence within certain Hilbert subspace had been discovered.
These studies spurred the development of the theory of decoherence-free-subspace \cite{lidar_dfs_prl_1998,lidar_review_dfs_2014}, which is conceived
as a promising solution to circumvent the obstacle of decohernece in quantum information science. The phenomenon of coherence-preserving can be observed in our
paradigm as well and is related to the delta component $\wp_6(x_6)=\delta(x_6)$ on $x_6$. Consequently, our procedure provides a potential application in
the detection of decoherence-free-subspace in terms of delta components in the (quasi-)distribution.

\begin{figure*}[th]
\includegraphics[width=\textwidth]{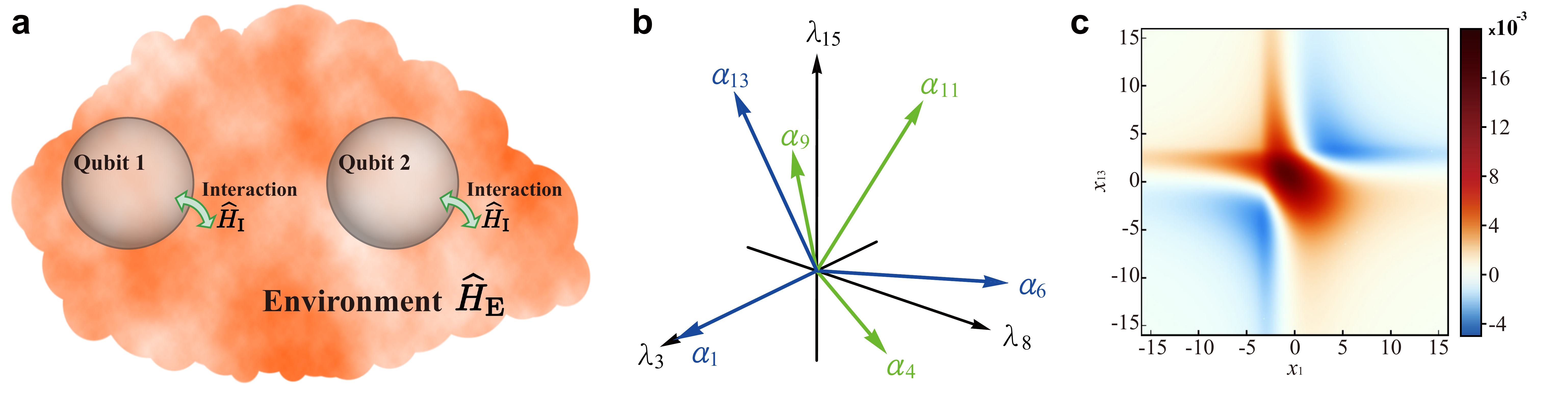}
\caption{Nonclassicality of the qubit pair pure dephasing. \textbf{a} A schematic illustration of our extended spin-boson model, describing a pair of non-interacting qubits coupled
to a common boson environment. \textbf{b} To simulate the qubit pair pure dephasing, $\wp(\boldsymbol{\lambda})$ is a (quasi-)distribution over $\mathbb{R}^3$ space spanned by $\lambda_3$,
$\lambda_8$, and $\lambda_{15}$. Here we show the six positive root vectors of $\mathfrak{su}(4)$. Three simple root vectors (blue) define a new set of random variables. The other
three non-simple root vectors (green) can be expressed as a combination of simple ones, e.g., $\boldsymbol{\alpha}_9=\boldsymbol{\alpha}_1+\boldsymbol{\alpha}_6+\boldsymbol{\alpha}_{13}$.
\textbf{c} The function $\wp_{1,13}(x_1,x_{13})$ distributes over the plane spanned by $x_1$ and $x_{13}$. For the case of Ohmic spectral density in the zero-temperature limit and
$\omega_\mathrm{c}=1$, it shows manifest negative regions and therefore indicates the nonclassicality of the qubit pair pure dephasing.}
\label{fig_ex_sbm}
\end{figure*}

\begin{figure*}[th]
\includegraphics[width=\textwidth]{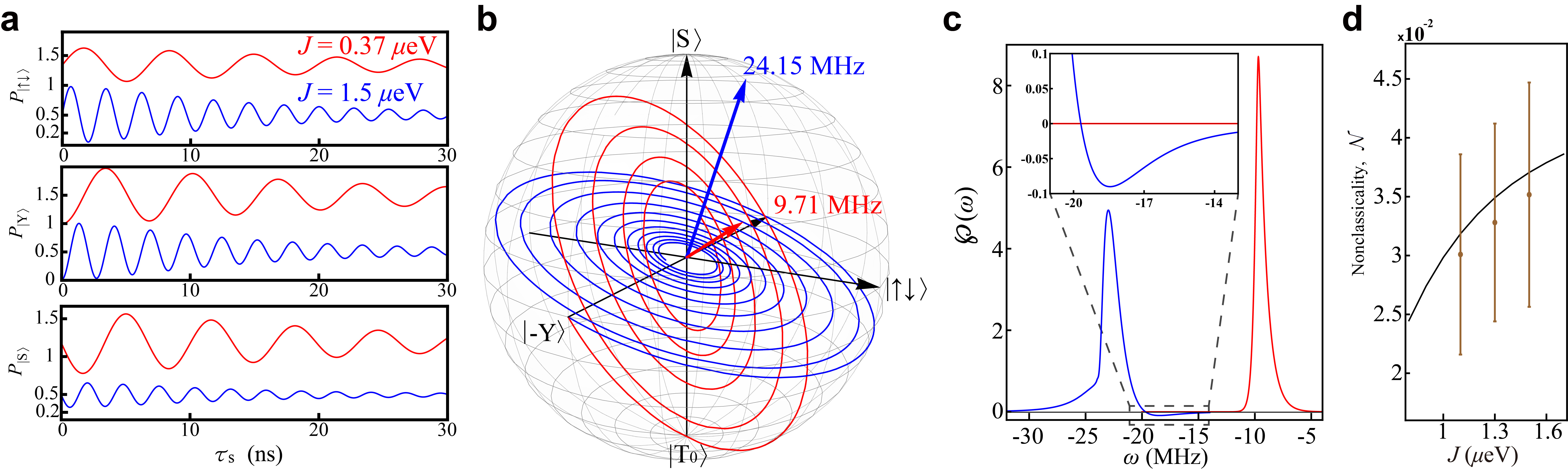}
\caption{Numerical simulation of the S-T$_0$ qubit pure dephasing. \textbf{a} The return probabilities $P_{|j\rangle}(\tau_\mathrm{s})$ are measured by projecting the states onto each axis after
a free induction decay time $\tau_\mathrm{s}$. Here we show two numerical simulations at different $J$ values. \textbf{b} The trajectories can be depicted in the Bloch sphere and the dynamics are therefore clearly visualized. The axes of rotation, as well as
the angle $\Omega$ between the $|\mathrm{S}\rangle$-axis, are identified by the normal vectors. \textbf{c} According to the rotation axes identified in (\textbf{b}), a unitary rotation
$\widehat{R}_\Omega$ recovers the standard form in Eq.~(\ref{eq_he_pure_dephasing}). Then our procedure is applicable. The resulting $\wp(\omega)$'s reflect
several physical intuitions, as explained in the main text. \textbf{d} The corresponding nonclassicality $\mathcal{N}$ at different $J$ can be estimated according to
Eq.~(\ref{eq_measure_nonclassicality}). It increases with $J$ in our simulation. In line with a realistic experimental modeling, statistically fluctuating noise is taken in account. The average
nonclassicalities (brown dots) are reduced due to the noise. The brown error bars are the standard deviations of the series of nonclassicality $\mathcal{N}$ values obtained
by repeatedly performing the noise simulation. More details are given in  Supplementary Note 9.}
\label{fig_ST0_qubit}
\end{figure*}

\noindent\textbf{Proposed experimental realization.} Finally, to underpin the practical feasibility of our approach, here we explain how to recover the dynamical
linear map $\mathcal{E}^{(\widetilde{L})}_t$ from the measurable $\chi$ matrix, which is a typical way to characterize arbitrary dynamics. The matrix elements
$\chi_{l,m}(t)$ are measured following the quantum process tomography technique, which has been applied in various architectures, e.g., optics
\cite{obrien_qpt_prl_2004,kiesel_qpt_prl_2005,pogorelov_qpt_pra_2017}, trapped ions \cite{riebe_qpt_prl_2006,monz_qpt_prl_2009}, and
superconductors \cite{bialczak_qpt_np_2010,yamamoto_qpt_prb_2010}.

Note that $\mathcal{E}^{(\widetilde{L})}_t$ on the left hand side of Eq.~(\ref{eq_general_f-s_transform}) describes the complete time evolution of the system, i.e., we need
to generate raw data of $\chi_{l,m}(t)$ as a time sequence. While this implies repeating the experiment for different time intervals, it does in principle not impose additional
technical difficulty. Finally, $\mathcal{E}^{(\widetilde{L})}_t$ can be reconstructed by combining the measured $\chi_{l,m}(t)$ (see \textbf{\textsf{Methods}}).

Here we also demonstrate a numerical simulation of the quantum state tomography experiment in the S-T$_0$ qubit \cite{petta_st0_qubit_science_2005,foletti_st0_qubit_nat_phys_2009}.
With spin relaxation on the order of milliseconds \cite{johnson_st0_relaxation_nature_2005}, the qubit dynamics is well approximated as pure dephasing
on the time scale of tens of nanoseconds. The qubit state is detected by measuring the return probabilities, i.e., projective measurements onto each axis of the Bloch sphere, after a
free induction decay time $\tau_\mathrm{s}$, as shown in Fig.~4\textbf{a} (see \textbf{\textsf{Methods}}).
With the measured return probabilities, we can depict the trajectories in the Bloch sphere (Fig.~4\textbf{b}).
This allows us to fully reconstruct the dynamics $\mathcal{E}_t$ of the qubit. Then applying our procedure outlined above, we can obtain the resulting $\wp(\omega)$ shown in Fig.~4\textbf{c}.
They reflect the fact that the $|\mathrm{S}\rangle$ possesses a lower eigenenergy than $|\mathrm{T}_0\rangle$, and the physical intuition that the shorter
the coherence time, the broader the $\wp(\omega)$.
Having recovered $\wp(\omega)$, the nonclassicality values can be estimated according to Eq.~(\ref{eq_measure_nonclassicality}).
To achieve realistic experimental conditions, we dress the theoretical model with statistical fluctuations (Fig.~4\textbf{d}). This confirms the robustness of the nonclassicality detection
against experimental errors (see Supplementary Note 9).
\vspace{12pt}

\noindent\textbf{\textsf{Discussion}}

\noindent
The studies on unveiling genuine quantum properties are very important since these discover the fundamental principle of nature and spur the growth of different
branches in physics and technologies. Particularly, in the field of quantum information science, highly quantum-correlated systems are critical resources for
prominent quantum information tasks which can hardly be accomplished efficiently by classical computers.

By genuine quantum properties, we refer to those that can never be resembled by classical strategies. For example, Bell's inequality is derived based on the
assumption of realism and locality, while the Wigner function explain a boson field in terms of classical phase space. Inspired by these works, our
characterization of process nonclassicality stems from the correspondence between the averaged dynamics of a HE and the dynamics reduced from a
system-environment arrangement \cite{hongbin_process_n_cla_prl_2018}.

By introducing the CHER, the role of classical strategy played by the simulating HE for a dynamics is even more apparent. The (quasi-)distribution is endowed with
an explanation in terms of a random rotation model. This also implies that the nonclassical properties of a dynamics can be well-characterized by a (quasi-)distribution.

Our main achievement here lies in the establishment of a constructive procedure to retrieve the (quasi-)distributions for pure dephasing of any dimension. Additionally,
along with the analysis of the underlying algebraic structure, we also achieve to prove its existence and uniqueness provided commuting member Hamiltonians.
Therefore, the CHER of pure dephasing is faithful. Accordingly, based on our studies, we propose a measure of nonclassicality of pure dephasing by comparing the
(quasi-)distributions in terms of variational distance. We also show that our measure is reasonable due to its convexity.

In order to make our procedure viable, we discuss how to implement our approach with the raw data measured by quantum process tomography.
Furthermore, we also demonstrate a numerical simulation of the S-T$_0$ qubit quantum state tomography, with which we implement our approach step by step.

Finally, let us remark that the generalization to the cases beyond pure dephasing or even nonunital dynamics invokes nonabelian algebraic structures. The
Baker-Campbell-Hausdorff formula is then required and therefore complicates the formulation here. On the other hand, our approach highlights an inherent difference between dephasing
and dissipative dynamics in terms of their underlying algebras. This may provide a new route toward the theory of open systems. Additionally, we also find that it would be interesting
to investigate how the notion of dynamical process nonclassicality is related to other quasi-distributions \cite{wootters_ann_phys_1987}.

($\equiv\widehat{\mathsf{\Phi}}\omega\widehat{\mathsf{\Phi}}\equiv$)$\sim$meow$\sim$ \newline

\vspace{12pt}

\noindent\textbf{\textsf{Methods}}

\noindent\textbf{Adjoint representation.} The adjoint representation is a particularly important tool in the theory of Lie algebra. It assigns each element in a
Lie algebra $\mathfrak{L}$ an endomorphism in $\mathfrak{gl}(\mathfrak{L})$ (i.e., a homomorphism from $\mathfrak{L}$ to itself) in terms of Lie bracket.
Therefore, $\mathfrak{gl}(\mathfrak{L})$ is a Lie algebra consisting of linear maps acting on $\mathfrak{L}$, wherein $\mathfrak{L}$ plays the role of a vector
space with the generators being its basis. The adjoint representation of each generator is constructed in terms of structure constants $c_{klm}$. See
Supplementary Note 1 for more details.

\noindent\textbf{Cartan subalgebra.} The structure of a Lie algebra $\mathfrak{L}$ is largely determined by its Lie bracket, i.e., the commutator acting on
$\mathfrak{L}$. A Lie algebra is said to be abelian if all its elements are mutually commutative. Let $\mathfrak{H}$ be a Lie subalgebra of $\mathfrak{L}$.
$\mathfrak{H}$ is said to be the CSA of $\mathfrak{L}$ if $\mathfrak{H}$ is the maximal abelian (and semisimple) subalgebra. A very important property is that,
for a Lie algebra consists of matrices, the elements in its CSA are all simultaneously diagonalizable for a suitably chosen basis.

In our case, to simulate pure dephasing dynamics, we deal with traceless and diagonal member Hamiltonians, taken from $\mathfrak{H}$ of $\mathfrak{su}(n)$. To be
noted, since the adjoint representation preserves the Lie bracket, the adjoint representation of $\mathfrak{H}$ is also a CSA of $\mathfrak{sl}(\mathfrak{u}(n))$.
However, even if $\mathfrak{H}$ is diagonal, its adjoint representation may not necessarily be diagonal as well since the generators of $\mathfrak{u}(n)$
are not the suitable basis for diagonalizing it.

\noindent\textbf{Root system.} For $n$-dimensional systems, there are $(n-1)$ generators in the $\mathfrak{H}$ of $\mathfrak{su}(n)$. Therefore, each member
Hamiltonian possesses $(n-1)$ parameters $\widehat{H}_{\boldsymbol{\lambda}}=\sum_{k=2}^n\lambda_{k^2-1}\widehat{L}_{k^2-1}/2$, with
$\{\widehat{L}_{k^2-1}\}_{k=2,3,\ldots,n}$ being the generators of $\mathfrak{H}$. Additionally, the $(n^2-n)$ roots $\boldsymbol{\alpha}_m$, associated to each root space
$\mathrm{span}\{\widehat{K}_m\}$, are $(n-1)$-dimensional vectors, forming the root system $\mathrm{R}=\{\boldsymbol{\alpha}_m\}_m$ of $\mathfrak{su}(n)$. Besides,
according to the theory of root space decomposition, the root system possesses the following critical properties: (1) the roots come in pairs in the sense that, if
$\boldsymbol{\alpha}_m$ is a root, then $-\boldsymbol{\alpha}_m$ is a root as well. This reduces the number of equations half since we are sufficient to consider the
positive roots alone. (2) Among the $(n^2-n)/2$ positive roots, $(n-1)$ simple roots provide the marginal of $\wp$ along different directions and the
others provide the information on the correlations between them. (3) For $\mathfrak{su}(n)$, the angle between any two non-pairing roots must be either $\pi/3$,
$\pi/2$, or $2\pi/3$. Furthermore, with the Fourier transform on groups, an $n$-dimensional pure dephasing is characterized by $(n^2-n)/2$ complex functions
$\phi_m(t)$, which are the dephasing factors associated to each root space $\mathrm{span}\{\widehat{K}_m\}$.

\noindent\textbf{Reconstructing $\mathcal{E}^{(\widetilde{L})}_t$ from the $\chi$ matrix.} In our approach, the reduced system dynamics is fully characterized by the dynamical
linear map $\mathcal{E}^{(\widetilde{L})}_t$, which is an $n^2\times n^2$ matrix acting on a state column vector $\rho=\{n^{-1},\boldsymbol{\rho}\}\in\mathbb{R}^{n^2}$.
On the other hand, in a quantum process tomography experiment, the dynamics is characterized by the measurable $\chi$ matrix representation, with the matrix elements defined
according to
\begin{equation}
\mathcal{E}_t\{\rho_0\}=\sum_{l,m=0}^{n^2-1}\chi_{l,m}(t)\widehat{L}_l\rho_0\widehat{L}_m.
\label{eq_chi_matrix_representation}
\end{equation}
Note that we have used the Hermiticity $\widehat{L}_m^\dagger=\widehat{L}_m$ in the above expression.

Now we demonstrate how to reconstruct $\mathcal{E}^{(\widetilde{L})}_t$ from the measured $\chi_{l,m}(t)$. For a given dynamics $\mathcal{E}_t$, the matrix elements
$[\mathcal{E}^{(\widetilde{L})}_t]_{jk}$ are defined by applying
\begin{equation}
\mathcal{E}_t\{\widehat{L}_k\}=\sum_{j=0}^{n^2-1}\widehat{L}_j[\mathcal{E}^{(\widetilde{L})}_t]_{jk}
\end{equation}
on each generator $\widehat{L}_k$. On the other hand, according to the measured $\chi_{l,m}(t)$ in Eq.~(\ref{eq_chi_matrix_representation}), we have
\begin{equation}
\mathcal{E}_t\{\widehat{L}_k\}=\sum_{l,m=0}^{n^2-1}\chi_{l,m}(t)\widehat{L}_l\widehat{L}_k\widehat{L}_m.
\end{equation}
From the above two equations, we can deduce that
\begin{equation}
[\mathcal{E}^{(\widetilde{L})}_t]_{jk}=\frac{1}{2}\sum_{l,m=0}^{n^2-1}\chi_{l,m}(t)\mathrm{Tr}\widehat{L}_j\widehat{L}_l\widehat{L}_k\widehat{L}_m,~j,k\neq0,
\end{equation}
\begin{equation}
[\mathcal{E}^{(\widetilde{L})}_t]_{j0}=\frac{1}{2}\sum_{l,m=0}^{n^2-1}\chi_{l,m}(t)\mathrm{Tr}\widehat{L}_j\widehat{L}_l\widehat{L}_m,~j\neq0,
\end{equation}
\begin{equation}
[\mathcal{E}^{(\widetilde{L})}_t]_{0k}=\frac{1}{n}\sum_{l,m=0}^{n^2-1}\chi_{l,m}(t)\mathrm{Tr}\widehat{L}_l\widehat{L}_k\widehat{L}_m,~k\neq0,
\end{equation}
and
\begin{equation}
[\mathcal{E}^{(\widetilde{L})}_t]_{00}=\chi_{0,0}(t)+\frac{2}{n}\sum_{l=1}^{n^2-1}\chi_{l,l}(t).
\end{equation}
In the above equations, we have used the facts that $\widehat{L}_0=\widehat{I}$ and $\mathrm{Tr}\widehat{L}^2_j=2$ for $j\neq0$.

\noindent\textbf{Recovering the S-T$_0$ trajectory from measured data.}
For a double-quantum-dot S-T$_0$ qubit, the three axes of the Bloch sphere are conventionally defined as $|\mathrm{X}\rangle=(|\mathrm{S}\rangle+|\mathrm{T}_0\rangle)/\sqrt{2}=|\uparrow\downarrow\rangle$,
$|\mathrm{Y}\rangle=(|\mathrm{S}\rangle-i|\mathrm{T}_0\rangle)/\sqrt{2}$, and $|\mathrm{Z}\rangle=|\mathrm{S}\rangle=(|\uparrow\downarrow\rangle-|\downarrow\uparrow\rangle)/\sqrt{2}$
being the singlet state, as shown in Fig.~4\textbf{b}. The free Hamiltonian in the S-T$_{0}$ basis is
\begin{equation}
\widehat{H}_{\mathrm{ST}_0}=\left[
\begin{array}{cc}
-J & g\mu_\mathrm{B}\Delta\mathrm{B}^z_{\mathrm{nuc}} \\
g\mu_\mathrm{B}\Delta\mathrm{B}^z_{\mathrm{nuc}} & 0
\end{array}
\right],
\end{equation}
where $J=0.37$ $\mu$eV (red) and $1.5$ $\mu$eV (blue) is the exchange energy between two dots, $\Delta \mathrm{B}^z_{\mathrm{nuc}}=10.5$ $\mathrm{mT}$ is the hyperfine
field gradient, $g=-0.44$ is the $g$-factor for GaAs, and $\mu_\mathrm{B}=57.8$ $\mu$eVT$^{-1}$ is Bohr's magneton.
Various kinds of initial states can be prepared with carefully designed pulse by controlling the voltage detuning between the quantum dots. After the initialization, the qubit undergoes a
free induction decay for a time period $\tau_\mathrm{s}$. Finally, projective measurements onto each axis are performed.

We numerically simulate the return probabilities $P_{|\uparrow\downarrow\rangle}(\tau_\mathrm{s})$, $P_{|\mathrm{Y}\rangle}(\tau_\mathrm{s})$, and $P_{|\mathrm{S}\rangle}(\tau_\mathrm{s})$
to each axis (Fig.~4\textbf{a}). Then the density matrix $\rho(\tau_\mathrm{s})=[\widehat{I}+\sum_{j=\mathrm{X},\mathrm{Y},\mathrm{Z}}r_j(\tau_\mathrm{s})\hat{\sigma}_j]/2$ can be determined by
\begin{equation}
r_j(\tau_\mathrm{s})=2P_{|j\rangle}(\tau_\mathrm{s})-1,~j=\mathrm{X},\mathrm{Y},\mathrm{Z}.
\end{equation}
And one can depict the trajectory $\mathbf{r}(\tau_\mathrm{s})=\{r_\mathrm{X}(\tau_\mathrm{s}),r_\mathrm{Y}(\tau_\mathrm{s}),r_\mathrm{Z}(\tau_\mathrm{s})\}$ in the Bloch sphere (Fig.~4\textbf{b}).
This helps us to identify the axis of rotation with bare rotation frequencies $\omega=\sqrt{J^2+(2g\mu_\mathrm{B}\Delta\mathrm{B}^z_{\mathrm{nuc}})^2}/\hbar$ and the angle $\Omega$
between the $|\mathrm{S}\rangle$-axis.

Finally, a unitary rotation $\widehat{R}_\Omega\rho(\tau_\mathrm{s})\widehat{R}_\Omega^\dagger$ with $\widehat{R}_\Omega=\exp[i\Omega\hat{\sigma}_\mathrm{Y}/2]$ recover the standard form
in Eq.~(\ref{eq_he_pure_dephasing}). Our procedure is then applicable and leads to
\begin{eqnarray}
&&\left[
\begin{array}{c}
1/2 \\
\hline
(r_\mathrm{X}(\tau_\mathrm{s})\cos\Omega-r_\mathrm{Z}(\tau_\mathrm{s})\sin\Omega)/2 \\
r_\mathrm{Y}(\tau_\mathrm{s})/2 \\
(r_\mathrm{X}(\tau_\mathrm{s})\sin\Omega+r_\mathrm{Z}(\tau_\mathrm{s})\cos\Omega)/2
\end{array}
\right]= \nonumber\\
&&\left[
\begin{array}{c|ccc}
1 & 0 & 0 & 0 \\
\hline
0 & \int\wp(\omega)\cos\omega\tau_\mathrm{s}d\omega & -\int\wp(\omega)\sin\omega\tau_\mathrm{s}d\omega & 0 \\
0 & \int\wp(\omega)\sin\omega\tau_\mathrm{s}d\omega & \int\wp(\omega)\cos\omega\tau_\mathrm{s}d\omega & 0 \\
0 & 0 & 0 & 1
\end{array}
\right]\cdot\left[
\begin{array}{c}
1/2 \\
\hline
0 \\
-1 \\
0
\end{array}
\right]. \nonumber\\
\end{eqnarray}
The numerical solutions are shown in Fig.~4\textbf{c}.

Further schematic illustration of the simulation and detailed analysis of the effects of noise are given in Supplementary Note 9.
\vspace{12pt}

\noindent\textbf{\textsf{Data availability}}

\noindent The data analyzed during the current study are available from the corresponding authors on reasonable request.
\vspace{12pt}


%
\vspace{12pt}

\noindent\textbf{\textsf{Acknowledgements}}

\noindent This work is supported partially by the National Center for Theoretical Sciences
and Ministry of Science and Technology, Taiwan, Grants No. MOST 107-2628-M-006-002-MY3, MOST 107-2627-E-006-001,
MOST 106-2811-M-006-044, MOST 107-2811-M-006-017, and MOST 107-2811-M-009-527,
and Army Research Office (Grant No. W911NF-19-1-0081).
J.B. is supported by an Institute of Information and Communications Technology Promotion (IITP) grant funded by the Korean government (MSIP) (Grant No. 2019-0-00831, EQGIS),
the KIST Institutional Program (2E29580-19-148), and ITRC Program(IITP2018-2019-0-01402).
F.N. is supported in part by the MURI Center for Dynamic Magneto-Optics via the
Air Force Office of Scientific Research (AFOSR) (FA9550-14-1-0040),
Army Research Office (ARO) (Grant No. W911NF-18-1-0358),
Asian Office of Aerospace Research and Development (AOARD) (Grant No. FA2386-18-1-4045),
Japan Science and Technology Agency (JST) (Q-LEAP program and CREST Grant No. JPMJCR1676),
Japan Society for the Promotion of Science (JSPS) (JSPS-RFBR Grant No. 17-52-50023, and JSPS-FWO Grant No. VS.059.18N),
RIKEN-AIST Challenge Research Fund,
and the John Templeton Foundation.
\vspace{12pt}

\noindent\textbf{\textsf{Author contributions}}

\noindent H.-B.C. conceived the research and carried out the calculations,
with help from P.-Y.L. and C.G., under the supervision of Y.-N.C.
J.B. proposed the idea of variational distance.
Y.-N.C. and F.N. were responsible for the integration among different research units.
All authors contributed to the discussion of the central ideas and to the manuscript.
\vspace{12pt}

\noindent\textbf{\textsf{Additional information}}

\noindent\textbf{Supplementary Information} accompanies this paper at https://doi.org/10.1038/s41467-019-11502-4.

\noindent\textbf{Competing interests:} The authors declare no competing interests.

\newpage
\clearpage

\newpage
\onecolumngrid
\clearpage
\makeatletter
\renewcommand*{\fnum@figure}{{\normalfont Supplementary Figure~\thefigure}}
\renewcommand{\thesection}{Supplementary Note \arabic{section}}
\setcounter{equation}{0}
\setcounter{figure}{0}
\setcounter{secnumdepth}{3}
\begin{center}
{\bf \Large Supplementary Information---Quantifying the nonclassicality of pure dephasing}
\end{center}

\noindent{\Large Chen et al.}

\newpage

\twocolumngrid

\section{MATHEMATICAL SUPPLEMENTS ON LIE ALGEBRA}

In this work, many results rely heavily on the techniques of Lie algebras. For the accessibility to a wide audience in physics, we provide some supplements on Lie
algebras.

\subsection{$\mathfrak{u}(n)$ Lie algebra}

Since both Hamiltonians $\widehat{H}_\lambda$ and density matrices $\rho$ are Hermitian, it is natural to deal with the problem in the space
$\mathfrak{u}(n)=\mathfrak{u}(1)\oplus\mathfrak{su}(n)$, which is spanned by the identity $\{\widehat{I}\}$ and $\{\widehat{L}_m\}_m$ of $n^2-1$ traceless Hermitian
generators, respectively. Every member Hamiltonian $\widehat{H}_\lambda$ is an element in $\mathfrak{u}(n)$, and can be expressed as a linear
combination of the generators
\begin{equation}
\widehat{H}_\lambda=\lambda_0\widehat{I}+\sum_{m=1}^{n^2-1} \lambda_m\widehat{L}_m=\lambda_0\widehat{I}+\boldsymbol{\lambda}\cdot\widehat{\mathbf{L}},
\end{equation}
where $\lambda_0\in\mathbb{R}$ and $\boldsymbol{\lambda}=\{\lambda_m\}_m\in\mathbb{R}^{n^2-1}$. Namely, $\lambda=\{\lambda_0,\boldsymbol{\lambda}\}\in\mathbb{R}^{n^2}$
parametrizes the member Hamiltonian $\widehat{H}_\lambda$. Additionally, since $\mathfrak{u}(1)$ commutes with $\mathfrak{su}(n)$
(i.e., $[\mathfrak{u}(1),\mathfrak{su}(n)]=[\lambda_0\widehat{I},\boldsymbol{\lambda}\cdot\widehat{\mathbf{L}}]=0$
$\forall\lambda_0\in\mathbb{R},\boldsymbol{\lambda}\in\mathbb{R}^{n^2-1}$), this renders $\lambda_0$ playing no role in each single realization of the unitary evolution:
\begin{equation}
\exp[-i\widehat{H}_\lambda t]\rho\exp[i\widehat{H}_\lambda t]=\exp[-i\boldsymbol{\lambda}\cdot\widehat{\mathbf{L}}t]\rho\exp[i\boldsymbol{\lambda}\cdot\widehat{\mathbf{L}}t].
\label{sup_eq_single_realization_uni_evo}
\end{equation}
Therefore, we first consider $\mathfrak{su}(n)$, and the space $\mathfrak{u}(1)$ can be easily included latter.

In Lie algebras, $\mathfrak{su}(n)$ itself is a vector space, and equipped with a bilinear Lie bracket
\begin{equation}
[\quad,\quad]:\mathfrak{su}(n)\times\mathfrak{su}(n)\rightarrow\mathfrak{su}(n),
\end{equation}
satisfying the following properties
\begin{enumerate}
\item $[\widehat{H}_{\boldsymbol{\lambda}},\widehat{H}_{\boldsymbol{\lambda}}]=0$,~$\forall\widehat{H}_{\boldsymbol{\lambda}}\in\mathfrak{su}(n)$.
\item $[\widehat{H}_1,[\widehat{H}_2,\widehat{H}_3]]+[\widehat{H}_2,[\widehat{H}_3,\widehat{H}_1]]+[\widehat{H}_3,[\widehat{H}_1,\widehat{H}_2]]=0$, $\forall\widehat{H}_{\boldsymbol{\lambda}}\in\mathfrak{su}(n)$.
\end{enumerate}
The Lie bracket largely determines the structure of a Lie algebra. This can be understood by applying it to the generators. For $\mathfrak{su}(n)$, the generators
satisfy
\begin{equation}
[\widehat{L}_k,\widehat{L}_l]=i2c_{klm}\widehat{L}_m
\end{equation}
and the $c_{klm}$'s are called the structure constants, which satisfy
\begin{equation}
c_{klm}=-c_{lkm}=-c_{mlk},
\label{eq_structure_constants}
\end{equation}
for $\mathfrak{su}(n)$.

\subsection{Representation}

To acquire further insight of an abstract Lie algebra, one seminal approach is to link it to another easier one; meanwhile, its algebraic structure can be
preserved. This can be achieved by introducing the concepts of homomorphism and representation.
\begin{defi}[Lie algebra homomorphism]
Let $\mathfrak{L}$ and $\mathfrak{L}^\prime$ be two Lie algebras over the same field $\mathcal{F}$. A linear map $f:\mathfrak{L}\rightarrow \mathfrak{L}^\prime$
is a homomorphism if it preserves the Lie brackets:
\begin{equation}
f([\widehat{H}_1,\widehat{H}_2])=[f(\widehat{H}_1),f(\widehat{H}_2)], \forall\widehat{H}_\lambda\in\mathfrak{L}.
\end{equation}
A homomorphism is an isomorphism, if it is injective and surjective in the sense of linear maps.
\end{defi}

\begin{defi}[Representation of a Lie algebra]
Let $\mathfrak{L}$ be a Lie algebra over a field $\mathcal{F}$. A representation of $\mathfrak{L}$ is a Lie algebra homomorphism $f$
\begin{equation}
f:\mathfrak{L}\rightarrow \mathfrak{gl}(\mathcal{V}),
\end{equation}
where $\mathfrak{gl}(\mathcal{V})$ is the general linear algebra of endomorphisms on the vector space $\mathcal{V}$.
\end{defi}

Therefore a representation $f$ assigns each $\widehat{H}_\lambda\in\mathfrak{L}$ an endomorphism $f(\widehat{H}_\lambda):\mathcal{V}\rightarrow\mathcal{V}$,
depending linearly on $\widehat{H}_\lambda$ and preserving Lie brackets.

\subsection{Adjoint representation}

A particularly important representation in the Lie algebra theory is the adjoint representation
\begin{equation}
\mathrm{ad}:\mathfrak{L}\rightarrow\mathfrak{gl}(\mathfrak{L}),
\end{equation}
with
\begin{equation}
\mathrm{ad}:\widehat{H}_\lambda\mapsto\widetilde{H}_\lambda=[\widehat{H}_\lambda,\quad].
\end{equation}
In other words, the adjoint representation conceives each $\widehat{H}_\lambda\in\mathfrak{L}$ as an endomorphism
$\mathrm{ad}\widehat{H}_\lambda=\widetilde{H}_\lambda$ acting on $\mathfrak{L}$, and its action is implemented by the Lie bracket
$\widetilde{H}_\lambda(\widehat{H}_{\lambda^\prime})=[\widehat{H}_\lambda,\widehat{H}_{\lambda^\prime}]$.

Since a Lie algebra $\mathfrak{L}$ itself is a vector space, this allows one to express each element $\widetilde{H}_\lambda\in\mathfrak{gl}(\mathfrak{L})$ in
terms of a matrix with respect to the generator of $\mathfrak{L}$. For $\mathfrak{su}(n)$, the adjoint representation $\widetilde{L}_m$ of each generator
$\widehat{L}_m$ is constructed in terms of structure constants $c_{klm}$.

For example, one generically takes the generators of $\mathfrak{su}(2)$ to be the Pauli matrices, which satisfy the commutation relation
\begin{equation}
[\hat{\sigma}_k,\hat{\sigma}_l]=i2\varepsilon_{klm}\hat{\sigma}_m
\end{equation}
cyclically. Therefore, the adjoint representation of the Pauli matrices are given by
\begin{eqnarray}
\tilde{\sigma}_x&=&\left[
\begin{array}{ccc}
0 & 0 & 0 \\
0 & 0 & -i2\\
0 & i2& 0
\end{array}
\right], \\
\tilde{\sigma}_y&=&\left[
\begin{array}{ccc}
0 & 0 & i2 \\
0 & 0 & 0 \\
-i2& 0 & 0
\end{array}
\right], \\
\tilde{\sigma}_z&=&\left[
\begin{array}{ccc}
0 &-i2& 0 \\
i2& 0 & 0 \\
0 & 0 & 0
\end{array}
\right].
\end{eqnarray}
And any element $\widehat{H}_{\boldsymbol{\lambda}}=\lambda_x\hat{\sigma}_x+\lambda_y\hat{\sigma}_y+\lambda_z\hat{\sigma}_z\in\mathfrak{su}(2)$ has a represention
\begin{equation}
\widetilde{H}_{\boldsymbol{\lambda}}=\boldsymbol{\lambda}\cdot\tilde{\boldsymbol{\sigma}}=\left[
\begin{array}{ccc}
0 & -i2\lambda_z & i2\lambda_y \\
i2\lambda_z & 0 & -i2\lambda_x\\
-i2\lambda_y & i2\lambda_x& 0
\end{array}
\right].
\end{equation}

Since $\mathfrak{u}(1)$ commutes with $\mathfrak{su}(2)$, one can easily extend the representation to the space
$\mathfrak{u}(2)=\mathfrak{u}(1)\oplus\mathfrak{su}(2)$, such that for any $\widehat{H}_\lambda=\lambda_0\widehat{I}+\boldsymbol{\lambda}\cdot\hat{\boldsymbol{\sigma}}\in\mathfrak{u}(2)$,
its adjoint representation is explicitly written as
\begin{equation}
\widetilde{H}_\lambda=\left[
\begin{array}{c|ccc}
0 & 0 & 0 & 0 \\
\hline
0 & 0 & -i2\lambda_z & i2\lambda_y \\
0 & i2\lambda_z & 0 & -i2\lambda_x\\
0 & -i2\lambda_y & i2\lambda_x& 0
\end{array}
\right].
\label{eq_adj_rep_su2}
\end{equation}
Notice that $\widetilde{H}_\lambda$ is independent of $\lambda_0$. This reflects the fact that each single unitary evolution in
Supplementary Equation~(\ref{sup_eq_single_realization_uni_evo}) has no $\lambda_0$ dependence.

Similarly, for the general cases, every $\widehat{H}_\lambda=\lambda_0\widehat{I}+\boldsymbol{\lambda}\cdot\widehat{\mathbf{L}}\in\mathfrak{u}(n)$ has a representation
\begin{equation}
\widetilde{H}_\lambda=\lambda_0\widetilde{I}+\boldsymbol{\lambda}\cdot\widetilde{\mathbf{L}}=\left[
\begin{array}{c|ccc}
0 & 0 & \cdots & 0  \\
\hline
0 &   &        &    \\
\vdots & & \boldsymbol{\lambda}\cdot\widetilde{\mathbf{L}} & \\
0 &   &        &
\end{array}
\right],
\end{equation}
which is also independent of $\lambda_0$.

\section{PROOF OF EQ.~(3) IN THE MAIN TEXT}

Here we present the translation from Eq.~(1) into Eq.~(3) in the main text. We begin with a useful tool.
\begin{lem}
Let $\widehat{L}$ and $\widehat{M}$ be any elements in the general linear group GL($n$) of $n\times n$ matrices. Then we have the following relation
\begin{equation}
\exp[\widehat{L}]\widehat{M}\exp[-\widehat{L}]=\sum_{\mu=0}^\infty\frac{1}{\mu!}[\widehat{L},\widehat{M}]_{(\mu)},
\end{equation}
where
\begin{eqnarray}
[\widehat{L},\widehat{M}]_{(0)}&=&\widehat{M}, \nonumber
\end{eqnarray}
\begin{eqnarray}
[\widehat{L},\widehat{M}]_{(1)}&=&[\widehat{L},\widehat{M}], \nonumber
\end{eqnarray}
\begin{eqnarray}
[\widehat{L},\widehat{M}]_{(\mu)}&=&[\widehat{L},[\widehat{L},\widehat{M}]_{(\mu-1)}].
\end{eqnarray}
\end{lem}
This lemma can be proven by straightforwardly expanding $\exp[\pm\widehat{L}]$ with its Taylor series. And an elementary algebra leads to the desired result.

With this lemma, a single realization of the unitary evolution in Eq.~(1) in the main text can be rewritten as
\begin{equation}
\exp[-i\widehat{H}_\lambda t]\rho\exp[i\widehat{H}_\lambda t]=\sum_{\mu=0}^\infty\frac{(-it)^\mu}{\mu!}[\widehat{H}_\lambda,\rho]_{(\mu)}.
\label{sup_eq_single_unitary_expansion}
\end{equation}
One can observe that the right hand side of Supplementary Equation~(\ref{sup_eq_single_unitary_expansion}) resembles the Taylor series of an exponential. To further recast it into
a closed exponential form, we must make use of the adjoint representation of the $\mathfrak{u}(n)$ Lie algebra we have discussed.

A density matrix $\rho$ is also Hermitian and of unital trace; it can be expressed in terms of $\rho=n^{-1}\widehat{I}+\boldsymbol{\rho}\cdot\widehat{\mathbf{L}}$, with
$\boldsymbol{\rho}\in\mathbb{R}^{n^2-1}$. One can conceive $\rho=\{n^{-1},\boldsymbol{\rho}\}$ as an $n^2$-dimensional column vector, then the action of the commutator
$[\widehat{H}_\lambda,\rho]$ can be expressed in terms of conventional matrix multiplication:
\begin{equation}
[\widehat{H}_\lambda,\rho]=\widetilde{H}_\lambda\cdot\rho=
\left[
\begin{array}{c|ccc}
0 & 0 & \cdots & 0  \\
\hline
0 &   &        &    \\
\vdots & & \boldsymbol{\lambda}\cdot\widetilde{\mathbf{L}} & \\
0 &   &        &
\end{array}
\right]
\cdot
\left[
\begin{array}{c}
n^{-1} \\
\hline
\\
\boldsymbol{\rho} \\
\\
\end{array}
\right],
\end{equation}
and therefore
\begin{equation}
[\widehat{H}_\lambda,\rho]_{(\mu)}=(\widetilde{H}_\lambda)^{\mu}\cdot\rho.
\end{equation}

Consequently, the exponential form of Supplementary Equation~(\ref{sup_eq_single_unitary_expansion}) follows immediately
\begin{eqnarray}
\exp[-i\widehat{H}_\lambda t]\rho\exp[i\widehat{H}_\lambda t]&=&\sum_{\mu=0}^\infty\frac{(-it)^\mu}{\mu!}(\widetilde{H}_\lambda)^{\mu}\cdot\rho \nonumber\\
&=&\exp[-i\widetilde{H}_\lambda t]\cdot\rho.
\end{eqnarray}
Then, given a time-independent HE $\{(p_\lambda,\widehat{H}_\lambda)\}$, it determines an unital and trace-preserving dynamical linear map
$\mathcal{E}^{(\widetilde{L})}_t$ via the Fourier transform on group:
\begin{equation}
\mathcal{E}^{(\widetilde{L})}_t=\int_{\mathbb{R}^{n^2}} p_\lambda e^{-i\widetilde{H}_\lambda t}d\lambda
=\int_{\mathbb{R}^{n^2}} p_\lambda e^{-i\lambda\widetilde{L}t}d\lambda,
\end{equation}
provided $\widetilde{H}_\lambda=\lambda_0\widetilde{I}+\boldsymbol{\lambda}\cdot\widetilde{\mathbf{L}}$. Notice that $\lambda=\{\lambda_0,\boldsymbol{\lambda}\}\in\mathbb{R}^{n^2}$ and
$\widetilde{I}=0$.

\section{CONVEXITY OF VARIATIONAL DISTANCE MEASURE}

We first note that the set of all legitimate probability distributions is convex since any statistical mixture of probability distributions is again a probability
distribution.

Suppose that we are given two pure dephasing dynamics $\mathcal{E}_t^1$ and $\mathcal{E}_t^2$ with $\wp_\lambda^1$ and $\wp_\lambda^2$ being their
(quasi-)distributions, respectively. According to our measure of nonclassicality, we have
\begin{eqnarray}
&&a\mathcal{N}\{\mathcal{E}_t^1\}+(1-a)\mathcal{N}\{\mathcal{E}_t^2\} \nonumber\\
&&=a\inf_{p_\lambda}\int_\mathcal{G}\frac{1}{2}\vert\wp^1_\lambda-p_\lambda\vert d\lambda
+(1-a)\inf_{p_\lambda}\int_\mathcal{G}\frac{1}{2}\vert\wp^2_\lambda-p_\lambda\vert d\lambda. \nonumber\\
\end{eqnarray}
Suppose that the two infimums are achieved by $p^1_\lambda$ and $p^2_\lambda$, respectively, we then have
\begin{eqnarray}
&&a\mathcal{N}\{\mathcal{E}_t^1\}+(1-a)\mathcal{N}\{\mathcal{E}_t^2\} \nonumber\\
&&=a\int_\mathcal{G}\frac{1}{2}\vert\wp^1_\lambda-p^1_\lambda\vert d\lambda+(1-a)\int_\mathcal{G}\frac{1}{2}\vert\wp^2_\lambda-p^2_\lambda\vert d\lambda \nonumber\\
&&\geq\int_\mathcal{G}\frac{1}{2}\vert a\wp^1_\lambda+(1-a)\wp^2_\lambda-[a p^1_\lambda+(1-a)p^2_\lambda]\vert d\lambda \nonumber\\
&&\geq\inf_{p_\lambda}\int_\mathcal{G}\frac{1}{2}\vert a\wp^1_\lambda+(1-a)\wp^2_\lambda-p_\lambda\vert d\lambda \nonumber\\
&&=\mathcal{N}\{a\mathcal{E}_t^1+(1-a)\mathcal{E}_t^2\}.
\end{eqnarray}
Therefore, our measure of nonclassicality is convex.

\section{FINDING THE CHER OF QUBIT PURE DEPHASING}

Within a properly chosen basis of its associated Hilbert space, any qubit pure dephasing dynamics can be expressed as
\begin{equation}
\rho_0=\left[
\begin{array}{cc}
\rho_{\uparrow\uparrow} & \rho_{\uparrow\downarrow} \\
\rho_{\downarrow\uparrow} & \rho_{\downarrow\downarrow}
\end{array}
\right]\mapsto
\mathcal{E}_t\{\rho_0\}=\left[
\begin{array}{cc}
\rho_{\uparrow\uparrow} & \rho_{\uparrow\downarrow}\phi(t) \\
\rho_{\downarrow\uparrow}\phi^\ast(t) & \rho_{\downarrow\downarrow}
\end{array}
\right].
\end{equation}
The diagonal elements are constant in time and the off-diagonal elements are governed by the dephasing factor $\phi(t)=\exp[-i\theta(t)-\Phi(t)]$, where
$\theta(t)$ ($\Phi(t)$) is a real odd (even) function on time $t$, respectively, such that $\phi(0)=1$, $\vert\phi(t)\vert\leq1$ for all $t\in\mathbb{R}$, and
$\phi(-t)=\phi^\ast(t)$. The first two conditions are for the complete positivity of the dynamics and the last one guarantees that the (quasi-)distribution $\wp$
is a real function.

If we expand $\rho$ in terms of $\rho=2^{-1}\widehat{I}+\boldsymbol{\rho}\cdot\hat{\boldsymbol{\sigma}}$, where $\hat{\boldsymbol{\sigma}}=\{\hat{\sigma}_x,\hat{\sigma}_y,\hat{\sigma}_z\}$
denotes three Pauli matrices, a qubit initial state can be expressed as a four-dimensional column vector
\begin{equation}
\rho_0=\left[
\begin{array}{c}
1/2 \\
\hline
(\rho_{\uparrow\downarrow}+\rho_{\downarrow\uparrow})/2 \\
i(\rho_{\uparrow\downarrow}-\rho_{\downarrow\uparrow})/2 \\
(\rho_{\uparrow\uparrow}-\rho_{\downarrow\downarrow})/2
\end{array}
\right].
\end{equation}

Now we know the action of $\mathcal{E}_t$ on a state, its linear map form $\mathcal{E}^{(\widetilde{\sigma})}_t$ can be constructed by applying it to the
generators:
\begin{enumerate}
\item $\mathcal{E}_t\{\widehat{I}\}=\widehat{I}$.
\item $\mathcal{E}_t\{\hat{\sigma}_x\}=e^{-\Phi(t)}\cos\theta(t)\hat{\sigma}_x+e^{-\Phi(t)}\sin\theta(t)\hat{\sigma}_y$.
\item $\mathcal{E}_t\{\hat{\sigma}_y\}=-e^{-\Phi(t)}\sin\theta(t)\hat{\sigma}_x+e^{-\Phi(t)}\cos\theta(t)\hat{\sigma}_y$.
\item $\mathcal{E}_t\{\hat{\sigma}_z\}=\hat{\sigma}_z$.
\end{enumerate}
We then have the dynamical linear map:
\begin{equation}
\mathcal{E}^{(\widetilde{\sigma})}_t=\left[
\begin{array}{c|ccc}
1 & 0 & 0 & 0 \\
\hline
0 & e^{-\Phi(t)}\cos\theta(t) & -e^{-\Phi(t)}\sin\theta(t) & 0 \\
0 & e^{-\Phi(t)}\sin\theta(t) & e^{-\Phi(t)}\cos\theta(t) & 0 \\
0 & 0 & 0 & 1
\end{array}
\right].\label{sup_eq_qubit_deph_dyn_map}
\end{equation}

On the other hand, the adjoint representation of $\hat{\sigma}_z$ (including the generator $\widehat{I}$ of $\mathfrak{u}(1)$) reads
\begin{equation}
\tilde{\sigma}_z=\left[
\begin{array}{c|ccc}
0 & 0 & 0 & 0 \\
\hline
0 & 0 &-i2& 0 \\
0 & i2& 0 & 0 \\
0 & 0 & 0 & 0
\end{array}
\right].
\end{equation}
The right-hand side of Eq.~(5) in the main text reads
\begin{eqnarray}
&&\int_{\mathbb{R}}\wp(\omega)e^{-i(\omega\tilde{\sigma}_z/2)t}d\omega= \nonumber\\
&&\left[
\begin{array}{c|ccc}
1 & 0 & 0 & 0 \\
\hline
0 & \int\wp(\omega)\cos\omega td\omega & -\int\wp(\omega)\sin\omega td\omega & 0 \\
0 & \int\wp(\omega)\sin\omega td\omega & \int\wp(\omega)\cos\omega td\omega & 0 \\
0 & 0 & 0 & 1
\end{array}
\right].\label{sup_eq_qubit_deph_he_adj_repr}
\end{eqnarray}

Note that the two matrices (\ref{sup_eq_qubit_deph_dyn_map}) and (\ref{sup_eq_qubit_deph_he_adj_repr}) can be simultaneously diagonalized by multiplying
\begin{equation}
X=\left[
\begin{array}{c|ccc}
1 & 0 & 0 & 0 \\
\hline
0 & 1 & -i & 0 \\
0 & 1 & i & 0 \\
0 & 0 & 0 & 1 \\
\end{array}
\right],
X^{-1}=\left[
\begin{array}{c|ccc}
1 & 0 & 0 & 0 \\
\hline
0 & \frac{1}{2} & \frac{1}{2} & 0 \\
0 & \frac{i}{2} & \frac{-i}{2} & 0 \\
0 & 0 & 0 & 1 \\
\end{array}
\right]
\end{equation}
from the left and the right, respectively. Namely,
\begin{equation}
X\cdot\mathcal{E}^{(\widetilde{\sigma})}_t\cdot X^{-1}=\left[
\begin{array}{c|ccc}
1 & 0 & 0 & 0 \\
\hline
0 & \phi(t) & 0 & 0 \\
0 & 0 & \phi^\ast(t) & 0 \\
0 & 0 & 0 & 1
\end{array}
\right]
\end{equation}
and
\begin{eqnarray}
&&\int_{\mathbb{R}}\wp(\omega)X\cdot e^{-i(\omega\tilde{\sigma}_z/2)t}\cdot X^{-1}d\omega= \nonumber\\
&&\left[
\begin{array}{c|ccc}
1 & 0 & 0 & 0 \\
\hline
0 & \int_\mathbb{R}\wp(\omega)e^{-i\omega t}d\omega & 0 & 0 \\
0 & 0 & \int_\mathbb{R}\wp(\omega)e^{i\omega t}d\omega & 0 \\
0 & 0 & 0 & 1
\end{array}
\right].
\end{eqnarray}
Therefore, the same conclusion
\begin{equation}
\exp[-i\theta(t)-\Phi(t)]=\int_\mathbb{R}\wp(\omega)e^{-i\omega t}d\omega \label{sup_eq_wp_f-s_transform}
\end{equation}
is immediately manifest and the conventional inverse Fourier transform leads to the final result.

\section{DIGONALIZATION AND ITS IMPLICATION}

In view of Supplementary Equations~(\ref{sup_eq_qubit_deph_dyn_map}) and (\ref{sup_eq_qubit_deph_he_adj_repr}), we can easily obtain the result Supplementary Equation~(\ref{sup_eq_wp_f-s_transform})
without diagonalizing them. Diagonalization seems not necessary. However, the diagonalization provides a deeper insight into the intrinsic algebraic structure.
It is essential for a systematic procedure when tackling higher dimensional problems.

To understand the implications of the diagonalization, we recall that, in the adjoint representation $\mathfrak{sl}(\mathfrak{u}(2))$, the Lie algebra
$\mathfrak{u}(2)$ plays the role of a vector space with the bases $\{\widehat{I},\hat{\sigma}_x,\hat{\sigma}_y,\hat{\sigma}_z\}$. The transformation described by
$X$ and $X^{-1}$ transforms the bases into $\{\widehat{I},\hat{\sigma}_+,\hat{\sigma}_-,\hat{\sigma}_z\}$, where
$\hat{\sigma}_\pm=(\hat{\sigma}_x\pm i\hat{\sigma}_y)/2$, which are the bases of $\mathfrak{gl}(2)=\mathfrak{u}(1)\oplus\mathfrak{sl}(2)$.

On the other hand, as seen in Supplementary Equation~(\ref{sup_eq_single_realization_uni_evo}), $\lambda_0$ is irrelevant in describing the dynamics. We therefore consider only the
traceless member Hamiltonian taken from $\mathfrak{H}$ of $\mathfrak{su}(2)$, namely, $\widehat{H}_\omega=\omega\hat{\sigma}_z/2$. The factor $2$ is included for
later convenience. Its adjoint representation with respect to $\mathfrak{gl}(2)$ basis is obtained by applying $\widetilde{H}_\omega$ on them; namely,
$\widetilde{H}_\omega(\hat{\sigma}_\pm)=[\omega\hat{\sigma}_z/2,\hat{\sigma}_\pm]=\pm1\cdot\omega\hat{\sigma}_\pm$ and
$\widetilde{H}_\omega(\hat{\sigma}_z)=[\omega\hat{\sigma}_z/2,\hat{\sigma}_z]=0$.
Its matrix form is written as
\begin{equation}
\widetilde{H}_\omega=\left[
\begin{array}{c|ccc}
0 & 0 & 0 & 0 \\
\hline
0 & \omega & 0 & 0 \\
0 & 0 &-\omega & 0 \\
0 & 0 & 0 & 0
\end{array}
\right].
\end{equation}
The operators $\{\widehat{I},\hat{\sigma}_+,\hat{\sigma}_-,\hat{\sigma}_z\}$ are the ``eigenvectors'' of $\widetilde{H}_\omega$ associated with the eigenvalues
$\{0,1,-1,0\}$, respectively. The eigenvalues $\pm1$ are therefore referred to as the roots (denoted by $\alpha_{1,2}$) associated to the root spaces
$\mathrm{span}\{\hat{\sigma}_\pm\}$, spanned by the operators $\hat{\sigma}_\pm$. For higher dimensional systems, the CSA $\mathfrak{H}$ possesses more
generators; namely, the member Hamiltonian contains more parameters than a single $\omega$. The roots are no longer real scalars but vectors in an Euclidean space.
This can be seen in the following example.

\section{ROOT SYSTEM}

The root space decomposition is a very important tool in the theory of Lie algebras, especially in describing the structure of an abstract Lie algebra, and has
many prominent applications in elementary particle physics and gauge field theory. However, to thoroughly understand this technique, we would encounter a divergent
bundle of mathematics. This would make it unaccessible to the wide audience in physics. From a practical viewpoint, we instead discuss the following qutrit example,
which demonstrates the core concept of the root space decomposition. This is enough for the scope of this work.

\subsection{Qutrit pure dephasing}

Consider a qutrit pure dephasing described by
\begin{equation}
\mathcal{E}_t\{\rho_0\}=\left[
\begin{array}{ccc}
\rho_{11} & \rho_{12}\phi_1(t) & \rho_{13}\phi_4(t) \\
\rho_{21}\phi_2(t) & \rho_{22} & \rho_{23}\phi_6(t) \\
\rho_{31}\phi_5(t) & \rho_{32}\phi_7(t) & \rho_{33}
\end{array}
\right].
\end{equation}
The ordering of the numbering of $\phi_m(t)$ is for the latter convenience. This will become clear in the following discussions. To guarantee the Hermicity of
$\rho(t)$, $\phi_1(t)=\phi_2^\ast(t)$ and so on. Moreover, they satisfy $\phi_m(0)=1$, $\vert\phi_m(t)\vert\leq1$ for all $t\in\mathbb{R}$, and
$\phi_m(-t)=\phi^\ast_m(t)$.

Inheriting from the Gell-Mann matrices, which form the conventional generators for $\mathfrak{su}(3)$, we define the generators for $\mathfrak{sl}(3)$ as follows:
\begin{eqnarray}
&&\widehat{K}_1=\widehat{K}_2^\dagger=\left[
\begin{array}{ccc}
0 & 1 & 0 \\
0 & 0 & 0 \\
0 & 0 & 0 \\
\end{array}
\right],
\widehat{K}_3=\widehat{L}_3=\left[
\begin{array}{ccc}
1 & 0 & 0 \\
0 & -1 & 0 \\
0 & 0 & 0 \\
\end{array}
\right], \nonumber\\
&&\widehat{K}_4=\widehat{K}_5^\dagger=\left[
\begin{array}{ccc}
0 & 0 & 1 \\
0 & 0 & 0 \\
0 & 0 & 0 \\
\end{array}
\right],
\widehat{K}_6=\widehat{K}_7^\dagger=\left[
\begin{array}{ccc}
0 & 0 & 0 \\
0 & 0 & 1 \\
0 & 0 & 0 \\
\end{array}
\right], \nonumber\\
&&\widehat{K}_8=\widehat{L}_8=\frac{1}{\sqrt{3}}\left[
\begin{array}{ccc}
1 & 0 & 0 \\
0 & 1 & 0 \\
0 & 0 & -2 \\
\end{array}
\right].
\end{eqnarray}
Additionally, $\widehat{K}_0=\widehat{I}$ is the generator for $\mathfrak{u}(1)$. Then, the dynamical linear map in this basis is a diagonalized matrix
\begin{equation}
\mathcal{E}^{(\widetilde{L})}_t=\left[\begin{array}{c|cccccccc}
1 & & & & & & & & \\
\hline
 &\phi_1(t)& & & & & & & \\
 & &\phi_2(t)& & & & & & \\
 & & & 1 & & & & & \\
 & & & &\phi_4(t)& & & & \\
 & & & & &\phi_5(t)& & & \\
 & & & & & &\phi_6(t)& & \\
 & & & & & & &\phi_7(t)& \\
 & & & & & & & & 1
\end{array}\right],
\end{equation}
which is obtained by applying the dynamics on each generator: $\mathcal{E}_t\{\widehat{K}_m\}=\phi_m(t)\widehat{K}_m$.

In general, a $3\times3$ Hermitian operator is a linear combination of the above 9 generators. However, as seen in Supplementary Equation~(\ref{sup_eq_single_realization_uni_evo}),
$\lambda_0$ is irrelevant in describing the dynamics. We therefore neglect $\lambda_0$ and consider only the traceless member Hamiltonians. Furthermore, since we
only consider the elements in $\mathfrak{H}$, the simulating HE is of the form
$\{(\wp(\lambda_3,\lambda_8),\widehat{H}_{\boldsymbol{\lambda}})\}_{\lambda_3,\lambda_8}$ with
$\widehat{H}_{\boldsymbol{\lambda}}=(\lambda_3\widehat{L}_3+\lambda_8\widehat{L}_8)/2\in\mathfrak{H}$ and $\boldsymbol{\lambda}=(\lambda_3,\lambda_8)\in\mathbb{R}^2$.
By estimating all the commutators $[\widehat{H}_{\boldsymbol{\lambda}},\widehat{K}_m]=(\boldsymbol{\alpha}_m\cdot\boldsymbol{\lambda})\widehat{K}_m$, we obtain its adjoint
representation in the $\mathfrak{gl}(3)$ basis
$\widetilde{H}_{\boldsymbol{\lambda}}=(\lambda_3\widetilde{L}_3+\lambda_8\widetilde{L}_8)/2=\mathrm{diag}\left[\begin{array}{c|ccccc}
0 & \lambda_3 & -\lambda_3 & 0 & (\lambda_3+\sqrt{3}\lambda_8)/2 & -(\lambda_3+\sqrt{3}\lambda_8)/2
\end{array}\right.$
$\left.\begin{array}{ccc}
(-\lambda_3+\sqrt{3}\lambda_8)/2 & -(-\lambda_3+\sqrt{3}\lambda_8)/2 & 0
\end{array}\right]$, being a diagonal matrix as well.


Finally, from Eq.~(4) in the main text, $\mathcal{E}^{(\widetilde{L})}_t=\int_\mathcal{G} p_\lambda e^{-i\widetilde{H}_\lambda t} d\lambda$, we conclude that the
(quasi-)distribution $\wp(\lambda_3,\lambda_8)$ is governed by the following simultaneous Fourier transforms:
\begin{equation}
\left\{
\begin{array}{l}
\phi_1(t)=\int_{\mathbb{R}^2}\wp(\lambda_3,\lambda_8)e^{-i\lambda_3t}d\lambda_3d\lambda_8 \\
\phi_4(t)=\int_{\mathbb{R}^2}\wp(\lambda_3,\lambda_8)e^{-i(\lambda_3+\sqrt{3}\lambda_8)t/2}d\lambda_3d\lambda_8 \\
\phi_6(t)=\int_{\mathbb{R}^2}\wp(\lambda_3,\lambda_8)e^{-i(-\lambda_3+\sqrt{3}\lambda_8)t/2}d\lambda_3d\lambda_8
\end{array}
\right..
\label{sup_eq_simu_qutrit_pure_deph}
\end{equation}

\subsection{Root system of $\mathfrak{su}(3)$}

Instead of being engaged in solving the Supplementary Equations~(\ref{sup_eq_simu_qutrit_pure_deph}), we look further insight into its structure in terms of the root system.
According to $\widetilde{H}(\lambda_3,\lambda_8)$ above, we can list all the roots of $\mathfrak{su}(3)$:
\begin{eqnarray}
\boldsymbol{\alpha}_1=-\boldsymbol{\alpha}_2=(1,0), \nonumber\\
\boldsymbol{\alpha}_4=-\boldsymbol{\alpha}_5=\left(\frac{1}{2},\frac{\sqrt{3}}{2}\right), \nonumber\\
\boldsymbol{\alpha}_6=-\boldsymbol{\alpha}_7=\left(-\frac{1}{2},\frac{\sqrt{3}}{2}\right).
\end{eqnarray}
They are two dimensional vectors of equal length on the $\lambda_3$-$\lambda_8$ plane. We plot them in Fig.~2 in the main text.

We can observe that the roots satisfy the following properties:
\begin{enumerate}[R1]
\item The roots come in pair, e.g., $\boldsymbol{\alpha}_1$ and $\boldsymbol{\alpha}_2$ are two roots pointing in opposite direction. The three roots $\boldsymbol{\alpha}_1$,
$\boldsymbol{\alpha}_4$, and $\boldsymbol{\alpha}_6$ are referred to be positive.
\item Among the three positive roots, $\boldsymbol{\alpha}_1$ and $\boldsymbol{\alpha}_6$ are simple and $\boldsymbol{\alpha}_4$ is not, since
    $\boldsymbol{\alpha}_4=\boldsymbol{\alpha}_1+\boldsymbol{\alpha}_6$. \label{sup_item_root_prop}
\item All the roots are of equal length and the angle between any two non-pairing roots is either $\pi/3$, $\pi/2$, or $2\pi/3$.
\end{enumerate}

Based on the observations, we can consider $\wp(\lambda_3,\lambda_8)$ as a distribution over the $\lambda_3$-$\lambda_8$ plane. Now we rewrite
\begin{equation}
\wp(\lambda_3,\lambda_8)d\lambda_3d\lambda_8=\wp^\prime(x_1,x_6)dx_1dx_6
\end{equation}
via the change of variables $x_m=\boldsymbol{\alpha}_m\cdot\boldsymbol{\lambda}$, $m=1,6$. Note that the Jacobian $\mathrm{Det}[\boldsymbol{\alpha}_1~\boldsymbol{\alpha}_6]^{-1}=2/\sqrt{3}$,
due to the change of variables has been absorbed into $\wp^\prime(x_1,x_6)$. The first and third lines in Supplementary Equations~(\ref{sup_eq_simu_qutrit_pure_deph}) lead to
\begin{equation}
\left\{
\begin{array}{l}
\phi_1(t)=\int_\mathbb{R}\wp_1(x_1)e^{-ix_1t}dx_1 \\
\phi_6(t)=\int_\mathbb{R}\wp_6(x_6)e^{-ix_6t}dx_6
\end{array}
\right..
\end{equation}
They are the marginals of $\wp$ along the directions $\boldsymbol{\alpha}_1$ and $\boldsymbol{\alpha}_6$, respectively. $\wp_1(x_1)$ and $\wp_6(x_6)$ can be obtained by
performing the inverse Fourier transform. Moreover, due to the property R\ref{sup_item_root_prop}, the second line in Supplementary Equations~(\ref{sup_eq_simu_qutrit_pure_deph})
describes the correlation between the new random variables $x_1$ and $x_6$. If we consider a special case, e.g., $\phi_4(t)=\phi_1(t)\phi_6(t)$, the second
equation implies that they are independent:
\begin{equation}
\wp^\prime(x_1,x_6)=\wp_1(x_1)\wp_6(x_6).
\end{equation}
This finishes solving $\wp$. For the case of correlated random variables, we consider an example of four-dimensions in the following section.

\section{QUBIT PAIR PURE DEPHASING}

We proceed with a non-trivial example in the presence of correlations between random variables. With this example, we can illustrate the intrinsic complexity of
the retrieval of (quasi-)distributions.

We consider the extended spin-boson model consisting of a non-interacting qubit pair coupled to a common boson bath. The total Hamiltonian reads
\begin{eqnarray}
\widehat{H}_\mathrm{T}&=&\sum_{j=1,2}\frac{\omega_j}{2}\hat{\sigma}_{z,j}+\sum_{\mathbf{k}}\omega_{\mathbf{k}}\hat{b}_{\mathbf{k}}^\dagger\hat{b}_{\mathbf{k}} \nonumber\\
&&+\sum_{j,\mathbf{k}}\hat{\sigma}_{z,j}\otimes(g_{j,\mathbf{k}}\hat{b}_{\mathbf{k}}^\dagger+g_{j,\mathbf{k}}^\ast\hat{b}_{\mathbf{k}}).
\end{eqnarray}
The whole system evolves unitarily according to the unitary operator (in the interaction picture):
\begin{eqnarray}
\widehat{U}^\mathrm{I}(t)&=&\exp\left[i\sum_{\mathbf{k}}\widehat{Z}_{\mathbf{k}}\widehat{Z}_{\mathbf{k}}^\dagger\left(\frac{\omega_{\mathbf{k}}t-\sin\omega_{\mathbf{k}}t}{\omega_{\mathbf{k}}^2}\right)\right] \nonumber\\
&&\times\exp\left[\sum_{\mathbf{k}}\widehat{Z}_{\mathbf{k}}\alpha_{\mathbf{k}}(t)\hat{b}_{\mathbf{k}}^\dagger-\widehat{Z}_{\mathbf{k}}^\dagger\alpha_{\mathbf{k}}^*(t)\hat{b}_{\mathbf{k}}\right],
\end{eqnarray}
where $\widehat{Z}_{\mathbf{k}}=\sum_{j=1,2}g_{j,\mathbf{k}}\hat{\sigma}_{z,j}$ and
$\alpha_{\mathbf{k}}(t)=-i\int_0^t e^{i\omega_{\mathbf{k}}\tau}d\tau=\left(1-e^{i\omega_{\mathbf{k}}t}\right)/\omega_{\mathbf{k}}$.

For simplicity, we assume that $g_{1,\mathbf{k}}=g_{2,\mathbf{k}}$. Tracing out the boson bath, the qubit pair pure dephasing is described by
\begin{equation}
\mathcal{E}_t\{\rho_0\}=\left[
\begin{array}{cccc}
\rho_{11} & \rho_{12}\phi_1(t) & \rho_{13}\phi_4(t) & \rho_{14}\phi_9(t) \\
\rho_{21}\phi_2(t) & \rho_{22} & \rho_{23}\phi_6(t) & \rho_{24}\phi_{11}(t) \\
\rho_{31}\phi_5(t) & \rho_{32}\phi_7(t) & \rho_{33} & \rho_{34}\phi_{13}(t) \\
\rho_{41}\phi_{10}(t) & \rho_{42}\phi_{12}(t) & \rho_{43}\phi_{14}(t) & \rho_{44}
\end{array}
\right],
\end{equation}
with dephasing factors
\begin{eqnarray}
\phi_1(t)&=&\phi_4(t)=\exp[i\theta(t)-\Phi(t)], \nonumber\\
\phi_6(t)&=&1, \nonumber\\
\phi_9(t)&=&\exp[-4\Phi(t)], \nonumber\\
\phi_{11}(t)&=&\phi_{13}(t)=\exp[-i\theta(t)-\Phi(t)],
\end{eqnarray}
where
\begin{eqnarray}
\theta(t)&=&4\int_0^\infty\frac{\mathcal{J}(\omega)}{\omega^2}(\omega t-\sin\omega t)d\omega, \nonumber\\
\Phi(t)&=&4\int_0^\infty\frac{\mathcal{J}(\omega)}{\omega^2}\coth\left(\frac{\hbar\omega}{2k_\mathrm{B}T}\right)(1-\cos\omega t)d\omega. \nonumber\\
\end{eqnarray}
And $\mathcal{J}(\omega)$ is the environmental spectral density function.

Following our procedure, to simulate the qubit pair pure dephasing, we consider the diagonalized member Hamiltonian taken from the CSA $\mathfrak{H}$ of
$\mathfrak{su}(4)$:
\begin{equation}
\widehat{H}_{\boldsymbol{\lambda}}=(\lambda_3\widehat{L}_3+\lambda_8\widehat{L}_8+\lambda_{15}\widehat{L}_{15})/2.
\end{equation}
By estimating all the commutators $[\widehat{H}_{\boldsymbol{\lambda}},\widehat{K}_m]=(\boldsymbol{\alpha}_m\cdot\boldsymbol{\lambda})\widehat{K}_m$, for $m=1,2,\ldots,14$, with
$\widehat{K}_m$ being the generators of $\mathfrak{gl}(4)$, we can list all the root vectors of $\mathfrak{su}(4)$:
\begin{eqnarray}
\boldsymbol{\alpha}_1&=&-\boldsymbol{\alpha}_2=(1,0,0), \nonumber\\
\boldsymbol{\alpha}_4&=&-\boldsymbol{\alpha}_5=\left(\frac{1}{2},\frac{\sqrt{3}}{2},0\right), \nonumber\\
\boldsymbol{\alpha}_6&=&-\boldsymbol{\alpha}_7=\left(-\frac{1}{2},\frac{\sqrt{3}}{2},0\right), \nonumber\\
\boldsymbol{\alpha}_9&=&-\boldsymbol{\alpha}_{10}=\left(\frac{1}{2},\frac{1}{2\sqrt{3}},\sqrt{\frac{2}{3}}\right), \nonumber\\
\boldsymbol{\alpha}_{11}&=&-\boldsymbol{\alpha}_{12}=\left(-\frac{1}{2},\frac{1}{2\sqrt{3}},\sqrt{\frac{2}{3}}\right), \nonumber\\
\boldsymbol{\alpha}_{13}&=&-\boldsymbol{\alpha}_{14}=\left(0,-\frac{1}{\sqrt{3}},\sqrt{\frac{2}{3}}\right).
\end{eqnarray}
Note that the roots $\boldsymbol{\alpha}_1,\ldots,\boldsymbol{\alpha}_{7}$, lying on the $\lambda_3$-$\lambda_8$ plane, are the same as those of $\mathfrak{su}(3)$. The root
system of $\mathfrak{su}(4)$ is even more complicated. For visual clarity, we only show six positive roots in Fig.~3\textbf{b} in the main text. Moreover,
among the six positive roots, $\boldsymbol{\alpha}_1$, $\boldsymbol{\alpha}_6$, and $\boldsymbol{\alpha}_{13}$ are simple because other positive roots can be obtained by combining them,
e.g., $\boldsymbol{\alpha}_9=\boldsymbol{\alpha}_1+\boldsymbol{\alpha}_6+\boldsymbol{\alpha}_{13}$ and $\boldsymbol{\alpha}_{11}=\boldsymbol{\alpha}_6+\boldsymbol{\alpha}_{13}$.

From the equation $\mathcal{E}^{(\widetilde{L})}_t=\int_{\mathbb{R}^3} p(\boldsymbol{\lambda})e^{-i\widetilde{H}_{\boldsymbol{\lambda}}t} d^3\boldsymbol{\lambda}$, the
(quasi-)distribution $\wp(\boldsymbol{\lambda})$, over $\mathbb{R}^3$ space, is governed by six simultaneous Fourier transforms:
\begin{equation}
\phi_m(t)=\int_{\mathbb{R}^3}\wp(\boldsymbol{\lambda})e^{-i(\boldsymbol{\alpha}_m\cdot\boldsymbol{\lambda})t}d^3\boldsymbol{\lambda},~m=1,4,6,9,11,13.
\label{sup_eq_simu_qubit_pair_pure_deph}
\end{equation}
Generically, the three random variables are correlated. To solve the correlated $\wp$, we therefore perform the change of variables
$x_m=\boldsymbol{\alpha}_m\cdot\boldsymbol{\lambda}$, $m=1,6,13$, because they are simple and can be used to expand the other roots, and we rewrite
\begin{equation}
\wp(\lambda_3,\lambda_8,\lambda_{15})d\lambda_3d\lambda_8d\lambda_{15}=\wp^\prime(x_1,x_6,x_{13})dx_1dx_6dx_{13}.
\end{equation}
Note that the Jacobian $\mathrm{Det}[\boldsymbol{\alpha}_1~\boldsymbol{\alpha}_6~\boldsymbol{\alpha}_{13}]^{-1}=\sqrt{2}$ due to the change of variables has been absorbed
into $\wp^\prime$. Then, the three axes of $\wp^\prime$ are defined by the three simple roots.

Additionally, since $\phi_6(t)=1$, we can observe the following correspondence between the root vectors and the dephasing factors:
\begin{eqnarray}
\boldsymbol{\alpha}_1+\boldsymbol{\alpha}_6=\boldsymbol{\alpha}_4 &\leftrightarrow& \phi_1(t)\phi_6(t)=\phi_4(t), \nonumber\\
\boldsymbol{\alpha}_6+\boldsymbol{\alpha}_{13}=\boldsymbol{\alpha}_{11} &\leftrightarrow& \phi_6(t)\phi_{13}(t)=\phi_{11}(t).
\end{eqnarray}
This implies that $\wp^\prime=\wp_6(x_6)\wp_{1,13}(x_1,x_{13})$ is separated into two parties and they can be determined according to the set of equations:
\begin{eqnarray}
\phi_1(t)&=&\int_\mathbb{R}\wp_1(x_1)e^{-ix_1 t}dx_1, \nonumber\\
\phi_{13}(t)&=&\int_\mathbb{R}\wp_{13}(x_{13})e^{-ix_{13} t}dx_{13}, \nonumber\\
\phi_9(t)&=&\int_{\mathbb{R}^2}\wp_{1,13}(x_1,x_{13})e^{-ix_1 t}e^{-ix_{13} t}dx_1dx_{13}, \nonumber\\
1&=&\int_\mathbb{R}\wp_6(x_6)e^{-ix_6 t}dx_6.
\label{sup_eq_qubit_pair_equations}
\end{eqnarray}
The first and second line specify the marginals of $\wp_{1,13}(x_1,x_{13})$ along the directions $\boldsymbol{\alpha}_1$ and $\boldsymbol{\alpha}_{13}$, respectively;
meanwhile, the third line describes the correlation between them. The last line immediately leads to the result $\wp_6(x_6)=\delta(x_6)$.

Consider the Ohmic spectral density $\mathcal{J}(\omega)=\omega\exp(-\omega/\omega_\mathrm{c})$ in the zero-temperature limit, the dephasing factors can be
calculated explicitly:
\begin{eqnarray}
\phi_1(t)&=&e^{i\theta(t)-\Phi(t)}=\frac{\exp[i(4\omega_\mathrm{c}t-4\arctan(\omega_\mathrm{c}t))]}{(1+\omega^2_\mathrm{c} t^2)^2},  \nonumber\\
\phi_{13}(t)&=&e^{-i\theta(t)-\Phi(t)}=\frac{\exp[-i(4\omega_\mathrm{c}t-4\arctan(\omega_\mathrm{c}t))]}{(1+\omega^2_\mathrm{c} t^2)^2},  \nonumber\\
\phi_9(t)&=&e^{-4\Phi(t)}=\frac{1}{(1+\omega^2_\mathrm{c} t^2)^8}.
\end{eqnarray}
The two marginals are easily obtained by the conventional Fourier transform
\begin{eqnarray}
\wp_1(x_1)&=&\frac{1}{2\pi}\int_{-\infty}^\infty\phi_1(t)e^{ix_1t}dt \nonumber\\
&=&\left\{
\begin{array}{cl}
\frac{1}{6\omega^4_\mathrm{c}}(x_1+4\omega_\mathrm{c})^3e^{-\frac{x_1+4\omega_\mathrm{c}}{\omega_\mathrm{c}}} &, x_1\geq-4\omega_\mathrm{c} \\
0 &, x_1<-4\omega_\mathrm{c}
\end{array}
\right. \nonumber\\
\end{eqnarray}
and
\begin{eqnarray}
\wp_{13}(x_{13})&=&\frac{1}{2\pi}\int_{-\infty}^\infty\phi_{13}(t)e^{ix_{13}t}dt \nonumber\\
&=&\left\{
\begin{array}{cl}
0 &, x_{13}>4\omega_\mathrm{c} \\
\frac{-1}{6\omega^4_\mathrm{c}}(x_{13}-4\omega_\mathrm{c})^3e^{\frac{x_1-4\omega_\mathrm{c}}{\omega_\mathrm{c}}} &, x_{13}\leq 4\omega_\mathrm{c}
\end{array}
\right.. \nonumber\\
\end{eqnarray}

However, since $\phi_1(t)\phi_{13}(t)\neq\phi_9(t)$, this implies that the two random variables $x_1$ and $x_{13}$ are correlated and
$\wp_1(x_1)\wp_{13}(x_{13})\neq\wp_{1,13}(x_1,x_{13})$. A difficulty in solving $\wp_{1,13}$ lies in the fact that, in the third line of
Supplementary Equations~(\ref{sup_eq_qubit_pair_equations}), there are two random variables, but they are accompanied with the same time variable $t$.

Interestingly, this can easily be solved by a simple ansatz. Let
\begin{eqnarray}
\theta(t_1-t_{13})&=&4\omega_\mathrm{c}(t_1-t_{13})-4\arctan[\omega_\mathrm{c}(t_1-t_{13})], \nonumber\\
\tau(t_1,t_{13})&=&2^{\frac{8t_1t_{13}}{(t_1+t_{13})^2}}, \nonumber\\
\Psi(t_1,t_{13})&=&2\tau(t_1,t_{13})\mathrm{ln}\left[1+\omega^2_\mathrm{c}\frac{(t_1+t_{13})^2}{\tau(t_1,t_{13})}\right].
\end{eqnarray}
One can observe that $\tau(t,0)=\tau(0,t)=1$ and $\tau(t,t)=4$, then
\begin{eqnarray}
&&\exp[i\theta(t_1-t_{13})-\Psi(t_1,t_{13})]= \nonumber\\
&&\int\int_{-\infty}^\infty \wp_{1,13}(x_1,x_{13})e^{-ix_1 t_1}e^{-ix_{13} t_{13}}dx_1dx_{13}
\end{eqnarray}
simultaneously recovers the first three lines in Supplementary Equations~(\ref{sup_eq_qubit_pair_equations}); namely, $\{t_1=t,t_{13}=0\}$ recovers the first line, $\{t_1=0,t_{13}=t\}$
recovers the second line, and $\{t_1=t,t_{13}=t\}$ recovers the third line. Meanwhile, it is a conventional two-dimensional Fourier transform with distinct time
variables $t_1$ and $t_{13}$. Therefore, $\wp_{1,13}$ can be solved by
\begin{eqnarray}
&&\wp_{1,13}(x_1,x_{13})= \frac{1}{4\pi^2}\times\nonumber\\
&&\int\int_{-\infty}^\infty e^{i\theta(t_1-t_{13})-\Psi(t_1,t_{13})} e^{ix_1 t_1}e^{ix_{13} t_{13}}dt_1dt_{13}.
\end{eqnarray}
This concludes the solution of Supplementary Equations~(\ref{sup_eq_qubit_pair_equations}). The numerical result is shown in Fig.~3\textbf{c} in the main text.

\section{PROOF OF EXISTENCE AND UNIQUENESS}

After introducing our procedure, we now show the proof of the existence and uniqueness of the CHER for pure dephasing.
Since we deal with (quasi-)distribution functions $\wp(\lambda)$, which are real [$\wp(\lambda)\in\mathbb{R}$], normalized [$\int\wp(\lambda)d\lambda=1$], but not
necessarily positive, we start with the $L^1(\mathcal{G})$ space consisting of real functions defined on a locally compact and ablian group $\mathcal{G}$ (generated by CSA $\mathfrak{H}$)
such that their absolute values are Lebesgue integrable. Note that the $L^1(\mathcal{G})$ forms a vector space and is a super set of all (quasi-)distributions. Conversely, an arbitrary
element $f\in L^1(\mathcal{G})$ may not necessarily be normalized.

Besides the addition in $L^1(\mathcal{G})$, we further define a binary operation, the ``multiplication"
$\ast:L^1(\mathcal{G})\times L^1(\mathcal{G})\rightarrow L^1(\mathcal{G})$, in terms of the convolution:
\begin{eqnarray}
h(\lambda)&=&(f\ast g)(\lambda) \nonumber\\
&=&\int_\mathcal{G}f(\lambda-\xi)g(\xi)d\xi\in L^1(\mathcal{G}),~\forall~f,g\in L^1(\mathcal{G}). \nonumber\\
\end{eqnarray}
Equipped with this multiplication, $L^1(\mathcal{G})$ forms a Banach algebra. We assign the delta function $\delta(\lambda)$ the role of multiplicative
identity in the sense that
\begin{equation}
(f\ast\delta)(\lambda)=(\delta\ast f)(\lambda)=f(\lambda),~\forall~f\in L^1(\mathcal{G}).
\end{equation}

On the other hand, consider the adjoint representation $\mathfrak{gl}(\mathfrak{u}(n))$ in the basis of $\mathfrak{gl}(n)$, we define a subset
$\mathcal{D}\subset\mathfrak{gl}(\mathfrak{u}(n))$ consisting of diagonalized maps such that their entries satisfy the conditions:
\begin{enumerate}
\item Every $\mathcal{E}_t^{(\widetilde{L})}\in\mathcal{D}$ is diagonalized.
\item The entry of $\mathcal{E}_t^{(\widetilde{L})}$ corresponding to $\widehat{I}$ is a real constant $A\in\mathbb{R}$.
\item The entries of $\mathcal{E}_t^{(\widetilde{L})}$ corresponding to opposite root spaces are complex conjugate to each other, $\phi_{-\vec{\alpha}}(t)=\phi_{\vec{\alpha}}^\ast(t)$,
and satisfy $\phi_{\vec{\alpha}}(0)=1$ and $\phi_{\vec{\alpha}}(-t)=\phi_{\vec{\alpha}}^\ast(t)$.
\item The entries of $\mathcal{E}_t^{(\widetilde{L})}$ corresponding to the $\mathfrak{H}$ of $\mathfrak{sl}(n)$ are $1$.
\end{enumerate}
Note that $\mathcal{D}$ forms an abelian group and the set of all CPTP pure dephasing dynamical maps is its subset.

After identifying the algebraic structures, the Fourier transform $\mathcal{E}^{(\widetilde{L})}_t=\int_\mathcal{G}\wp(\lambda)e^{-i\lambda\widetilde{L}t} d\lambda$ can
be conceived as a map: $\wp(\lambda)\mapsto\mathcal{E}_t^{(\widetilde{L})}$. Then, given $\widetilde{L}_m\in\mathfrak{H}$, the Fourier transform is an isomorphism from
$L^1(\mathcal{G})$ to $\mathcal{D}$. This is stated in the following critical Lemma:
\begin{lem}
The Fourier transform with generators $\widetilde{L}_m$ taken from the CSA $\mathfrak{H}$ of $\mathfrak{gl}(\mathfrak{u}(n))$ is an isomorphism from $L^1(\mathcal{G})$
to $\mathcal{D}$.
\end{lem}
\begin{proof}
Suppose that $f$ and $g$ are two elements of $L^1(\mathcal{G})$, and $\mathcal{E}^{(\widetilde{L})}_f$ and $\mathcal{E}^{(\widetilde{L})}_g$ are their Fourier transform,
with generators $\widetilde{L}_m\in\mathfrak{H}$, respectively. Let $h=f\ast g$, then
\begin{eqnarray}
\mathcal{E}^{(\widetilde{L})}_h&=&\int_\mathcal{G}h(\lambda)e^{-i\lambda\widetilde{L}t}d\lambda \nonumber\\
&=&\int_\mathcal{G}(f\ast g)(\lambda)e^{-i\lambda\widetilde{L}t}d\lambda \nonumber\\
&=&\int_\mathcal{G}\int_\mathcal{G}f(\lambda-\xi)e^{-i\lambda\widetilde{L}t}g(\xi)d\xi d\lambda \nonumber\\
&=&\int_\mathcal{G}f(\lambda-\xi)e^{-i(\lambda-\xi)\widetilde{L}t} d\lambda\int_\mathcal{G}g(\xi)e^{-i\xi\widetilde{L}t}d\xi. \nonumber\\
\end{eqnarray}
The last line is valid with the following two properties. First, for the generators $\widetilde{L}_m\in\mathfrak{H}$, they commute with each other. Otherwise, we
must appeal to the BCH formula. Second, since $\mathcal{G}$ is an abelian group and $\lambda$ runs over all group elements, the rearrangement lemma
guarantees that $\lambda^\prime=\lambda-\xi$ is also the case.
Then, we have
\begin{equation}
\mathcal{E}^{(\widetilde{L})}_h=\mathcal{E}^{(\widetilde{L})}_f\mathcal{E}^{(\widetilde{L})}_g.
\end{equation}
Therefore, the Fourier transform is a multiplicative homomorphism from $L^1(\mathcal{G})$ to $\mathcal{D}$.

In addition, it is obvious that
\begin{equation}
\mathrm{id}^{(\widetilde{L})}=\int_\mathcal{G}\delta(\lambda)e^{-i\lambda\widetilde{L}t}d\lambda,~\forall~t\in\mathbb{R}.
\label{sup_eq_iden_f-s_transform}
\end{equation}
This means that the multiplicative identity $\delta(\lambda)$ in $L^1(\mathcal{G})$ is mapped to the identity map $\mathrm{id}^{(\widetilde{L})}$ in
$\mathcal{D}$, with all diagonal entries being $1$. Additionally, $\delta(\lambda)$ is the only element in the kernel of the Fourier transform
with $\widetilde{L}_m\in\mathfrak{H}$. Namely, $\delta(\lambda)$ is the only solution satisfying Supplementary Equation~(\ref{sup_eq_iden_f-s_transform}). This can easily be seen
from our procedure Eq.~(8) in the main text. Consequently, this proves our results that the Fourier transform
with $\widetilde{L}_m\in\mathfrak{H}$ is an isomorphism.
\end{proof}

This lemma ensures the one-one correspondence between $L^1(\mathcal{G})$ and $\mathcal{D}$. Moreover, a CPTP pure dephasing is an element in $\mathcal{D}$ with
$A=1$; this is equivalent to a normalized $\wp(\lambda)$. This proves the existence and uniqueness of the simulating HE with diagonalized member Hamiltonians for
CPTP pure dephasing.

\section{SIMULATING THE NOISE IN THE S-T$_0$ PURE DEPHASING EXPERIMENT}

Experiments inevitably suffer from the disturbances caused by the fluctuations of the surrounding environment or the imperfection of the measurements. Therefore,
the experimentally measured raw data may potentially deviate from the theoretical prediction of an idealized model.

The prototype of our theoretical model is the S-T$_0$ qubit in a gate-defined double-quantum-dot device, fabricated in a GaAs/AlGaAs heterostructure. The reported spin relaxation time (T$_1$)
in such material can approach several milliseconds, while the time-averaged dephasing time (T$_2^\ast$) is on the time scale of tens of nanoseconds. Therefore the qubit dynamics can be well
approximated as pure dephasing.

In the quantum state tomography experiment, the S-T$_0$ qubit state is constructed by projective measurements onto the three axes of the Bloch sphere defined as
$|\mathrm{X}\rangle=(|\mathrm{S}\rangle+|\mathrm{T}_0\rangle)/\sqrt{2}=|\uparrow\downarrow\rangle$, $|\mathrm{Y}\rangle=(|\mathrm{S}\rangle-i|\mathrm{T}_0\rangle)/\sqrt{2}$, and
$|\mathrm{Z}\rangle=|\mathrm{S}\rangle=(|\uparrow\downarrow\rangle-|\downarrow\uparrow\rangle)/\sqrt{2}$, and measuring the corresponding return probabilities $P_{|j\rangle}(\tau_\mathrm{s})$, $j=\mathrm{X},\mathrm{Y},\mathrm{Z}$, at different free induction decay times $\tau_\mathrm{s}$, as indicated by the blue curves in Supplementary Figure~1\textbf{a}. Once all the
$P_{|j\rangle}(\tau_\mathrm{s})$ are given, we can construct the density matrix
$\rho(\tau_\mathrm{s})=[\widehat{I}+\sum_{j=\mathrm{X},\mathrm{Y},\mathrm{Z}}r_j(\tau_\mathrm{s})\hat{\sigma}_j]/2$ with the trajectories
$\mathbf{r}(\tau_\mathrm{s})=\{r_\mathrm{X}(\tau_\mathrm{s}),r_\mathrm{Y}(\tau_\mathrm{s}),r_\mathrm{Z}(\tau_\mathrm{s})\}$ in the Bloch sphere determined by
\begin{equation}
r_j(\tau_\mathrm{s})=2P_{|j\rangle}(\tau_\mathrm{s})-1,~j=\mathrm{X},\mathrm{Y},\mathrm{Z}.
\label{sup_eq_bloch_sph_traj}
\end{equation}
Then we can apply the analysis explained in the main text and the \textbf{\textsf{Methods}} section.

We complement our theoretical simulation of the S-T$_0$ qubit pure dephasing by including noise effects in terms of statistical fluctuations. The detailed simulation of the noise
effects, as well as our complete analysis, are outlined step by step in the following, and schematically in Supplementary Figure~1.

\begin{figure*}[ht]
\includegraphics[width=\textwidth]{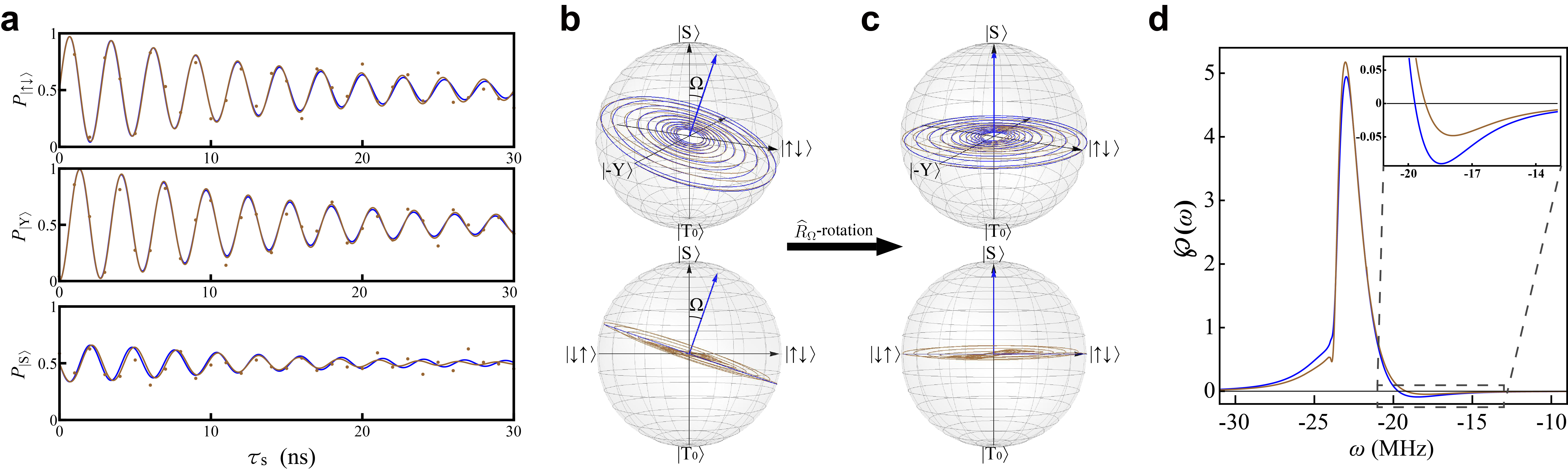}
\caption{\textbf{Schematic illustration of our numerical analysis of the noisy S-T$_0$ qubit pure dephasing.} \textbf{a} The blue curves stem from the theoretical model. The brown points are randomly offset vertically from the blue curves, following a Gaussian noise distribution with standard deviation 0.05. Then the
brown curves fitting the noisy data points simulate the noisy experimental measurement.
\textbf{b} After all the $P_{|j\rangle}(\tau_\mathrm{s})$ being given, we can depict the trajectories in a Bloch sphere and the dynamics are therefore explicitly
visualized. The theoretical (blue) trajectory defines a clear dephasing disk. Its normal vector and the angle $\Omega$ between the $|\mathrm{S}\rangle$-axis
can be identified. However, the noisy (brown) trajectory does not perfectly attach to the dephasing disk. The two panels are shown from different viewing angles. \textbf{c} According to the normal
vector identified in (\textbf{b}), a unitary rotation $\widehat{R}_\Omega$ recovers the standard form in Eq.~(2) in the main text. The two panels are shown from different viewing angles.
\textbf{d} Applying our procedure explained in the main text, we can numerically recover the desired $\wp(\omega)$.
Finally, repeatedly performing the noise simulation leads to a series of fluctuating $\mathcal{N}$ values. Then taking the mean value and the standard deviation, we obtain
the average nonclassicality $\mathcal{N}$ (brown points) and the brown error bars shown in Fig.~4\textbf{d} in the main text.}
\end{figure*}

\begin{enumerate}[Step 1]
\item Based on the theoretically simulation (blue curves), the brown points in Supplementary Figure~1\textbf{a} are randomly offset vertically, following a Gaussian noise
    distribution with standard deviation 0.05. Then the brown curves fitting the noisy data points simulate the noisy experimental measurement.

\item Depict the trajectories in the Bloch sphere according to Supplementary Equation~(\ref{sup_eq_bloch_sph_traj}) and identify the angle $\Omega$ between the axis of rotation (blue vector), i.e.,
    the normal vector of the blue dephasing disk, and $|\mathrm{S}\rangle$-axis, as shown in Supplementary Figure~1\textbf{b}.

\item Perform a unitary rotation $\widehat{R}_\Omega\rho(\tau_\mathrm{s})\widehat{R}_\Omega^\dagger$, with $\widehat{R}_\Omega=\exp[i\Omega\hat{\sigma}_\mathrm{Y}/2]$, as shown in Supplementary
    Figure~1\textbf{c}. This recovers the standard form Eq.~(2) in the main text.

\item For the idealized (blue) trajectory, our procedure is directly applicable and leads to the numerical result $\wp(\omega)$ in Supplementary Figure~1\textbf{d}. For the noisy (brown) trajectory,
    we first project it onto the dephasing disk. Then we can again apply our procedure.

\item Estimate the nonclassicality $\mathcal{N}$ according to Eq.~(4) in the main text, and repeatedly perform the noise simulation. This way, we can obtain a series of
    fluctuating nonclassicality $\mathcal{N}$ values. By taking the mean value and the standard deviation of the $\mathcal{N}$ series, we obtain the average
    nonclassicality $\mathcal{N}$ (brown points) and the brown error bars shown in Fig.~4\textbf{d} in the main text.

\end{enumerate}

\end{document}